\documentclass[12pt,titlepage]{article}
\pdfoutput=1
\usepackage{amsmath,amsthm,amssymb,mathtools}
\usepackage{graphicx}
\usepackage{array}
\usepackage{bbm,bm}
\usepackage{amsthm}
\usepackage{subfig}
\usepackage{rotating}
\usepackage{booktabs}
\usepackage{color}
\usepackage{natbib}
\usepackage[mathscr]{eucal}
\newtheorem{Theorem}{Theorem}[section]
\newtheorem{Lemma}[Theorem]{Lemma}

\theoremstyle{definition}
\newtheorem{Definition}[Theorem]{Definition}

\theoremstyle{remark}
\newtheorem{Remark}[Theorem]{Remark}

\numberwithin{equation}{section}
\numberwithin{figure}{section}
\numberwithin{table}{section}

\newcommand{\EE}{\operatorname{E}}
\newcommand{\PP}{\operatorname{P}}

\newcommand{\opt}{\mathrm{opt}}

\newcommand{\wvec}{\bm{w}}

\newcommand{\Qvec}{\bm{Q}}

\newcommand{\thetavec}{\bm{\theta}}

\newcommand{\Psivec}{\bm{\Psi}}

\newcommand{\ellvec}{\bm{\ell}}

\begin{document}

\title{MAX-FACTOR INDIVIDUAL RISK MODELS WITH APPLICATION TO CREDIT PORTFOLIOS}
\author{{\Large M}ICHEL {\Large D}ENUIT, {\Large A}NNA {\Large K}IRILIOUK$^\star$, {\Large J}OHAN {\Large S}EGERS\\
Institut de Statistique, Biostatistique et Sciences Actuarielles\\
Universit\'e catholique de Louvain\\
Louvain-la-Neuve, Belgium
\\[3mm]
$^\star$ Corresponding author: \texttt{anna.kiriliouk@uclouvain.be}}
\maketitle
\begin{abstract}
Individual risk models need to capture possible correlations as failing to do so typically results in an underestimation
of extreme quantiles of the aggregate loss. Such dependence modelling is particularly important for managing credit
risk, for instance, where joint defaults are a major cause of concern. Often, the dependence between the individual loss occurrence indicators is driven by a small
number of unobservable factors. 
Conditional loss probabilities are then expressed as monotone functions of linear combinations of these hidden factors.
However, combining the factors in a linear way allows for some compensation between them. Such diversification effects
are not always desirable and this is why the present work proposes a new model replacing linear combinations
with maxima. These max-factor models give more insight into which of the factors is dominant. 
\\[3mm]
\textit{Key words and phrases:} calibration, default indicator, dependence modelling, latent factors, loss occurrence.
\end{abstract}

\section{Introduction and motivation}

Individual risks are often exposed to the same environment and this induces some dependence that leads to bias in calculations
of  stop-loss premiums and other risk measures. 
There are many situations in practice where dependence affects occurrences of losses. 
Typical cases arise for policies covering natural disasters (hurricane, tornado, flood, etc.). 
We refer the reader to the book of \citet{denuit2006} for an introduction to the modelling
of dependence and to the review paper of \citet{anastasiadis2012} for
an overview of the various multivariate insurance models suggested in the literature.

In this paper, we model the occurrence of losses at the individual level. Recall that
portfolios of risks are generally described by means of either a bottom-up approach or
a top-down approach. In insurance, these two approaches are referred to as
the individual and the collective models of risk theory. The bottom-up approach is also
known as a name-per-name approach in the credit risk literature. It starts from a description of the
individual risks from which the distribution of the aggregate loss is derived. The bottom-up approach has some
clear advantages over the top-down approach, such as the possibility to easily
account for heterogeneity. 

In credit risk models, default indicators can in general not be considered as being mutually independent.
Dependence between the defaults  of different firms can be caused by direct links
between them (e.g., one firm is the other's largest customer) or by more
indirect links. In the latter category, we find industrial firms using the
same resources, and thus exposed to the same price shocks,
or selling on the same markets, and thus tributary of the same demand and subject to the same regulation.

A number of macroeconomic factors may influence many default indicators at once; examples include
including business cycles, level of unemployment, or shifts in monetary policy.
To account for these situations, vectors default indicators are often modelled via common mixture models.
The idea is that there exists a limited number of systematic factors  such that the
default indicators are conditionally independent when the factors are controlled.
Unconditionally, however, the default indicators are dependent because they are subject
to the same unobservable macroeconomic factors.
These factor models are among the few models that can replicate a realistic
correlated default behavior while dramatically reducing the numerical complexity
when computing the distribution of the aggregate portfolio loss.

In general, conditional default probabilities
are functions of linear combinations of the hidden factors, with weights
reflecting the relative sensitivity to the  risk factor.
This is the case for the majority of industry models, including the CreditRisk$^+$ and KMV models. 
We refer the reader to \citet{bluhm2002} for a general introduction. 
The hidden factors are typically associated to different levels of the economy in a hierarchical
way, accounting for global effects and sector-specific ones.
Replacing linear combinations of hidden factors with maxima
is attractive in some applications. The max-decomposition better accounts for shocks
specific to a given category of risks, whereas linear combinations of factors tend to
dilute the shock within the contributions of each factor to the sum. 

The remainder of this paper is organized as follows. In Section~\ref{sec:maxfactor}, we describe the proposed max-factor
specification to induce dependence between loss indicators. Section~\ref{sec:calibration} is devoted to calibration techniques. First, we describe the general setup and introduce some parametric factor models.
Second, we propose efficient numerical procedures to obtain the maximum likelihood estimates for max-factor models. In Section~\ref{sec:nonpar}, new nonparametric estimators are proposed that can be used
as benchmark to evaluate the goodness-of-fit of parametric risk models.
A simulation study assessing the performances of the estimators is given in Appendix B, whereas formal proofs
of the results proposed in Sections~\ref{sec:calibration} and \ref{sec:nonpar} can be found in Appendix A.
In Section~\ref{sec:numeric}, we work out a detailed numerical illustration performed on a classical credit risk data set provided in \citet{sp2001}. Finally, Section~\ref{sec:discuss} briefly discusses the results obtained in this paper and concludes.

\section{Max-factor risk model}\label{sec:maxfactor}

Consider a portfolio of $m$ risks split into $k$ categories observed over a given reference period. Each category, $r$, contains $m_r$ individual risks, $r=1,\ldots,k$. The indicator $Y_{r,i}$ is equal to 1 if
risk $i$ from category $r$ brings some financial loss and to 0 otherwise.
The random variables $Y_{r,i}$ may be associated to a borrower's default in credit risk,
to a policyholder's death in life insurance, or to the occurrence of a claim in general insurance, for instance.
Henceforth, we refer to $Y_{r,i}$ as the loss (occurrence) indicator.

As individual contracts are subject to a common environment, loss indicators are impacted
by a number of identical risk factors. The max-factor decomposition accounts for this
positive correlation by means of a global risk factor $\Psi_0$ affecting all the $m$
contracts and category-specific factors $\Psi_1,\ldots,\Psi_k$ whose influence is
restricted to the contracts in the same class. The random variables $\Psi_0,\Psi_1,\ldots,\Psi_k$
are assumed to be independent with common distribution function $F_\Psi$.
All the contracts in the same risk class $r$ share the common random effect $\Psi_r$ but are also subject to a competing global effect
$\Psi_0$ affecting the entire block of business. In homeowners insurance, this global effect may be related to storms
or earthquakes. In life insurance, it typically accounts for the sudden increase in death probabilities due to the
occurrence of pandemics.

Write $\Psivec=(\Psi_0,\Psi_1,\ldots,\Psi_k)$.
Whereas the majority of factor models are based on linear combinations of the hidden risk factors,
here we specify a latent-shock or competing-risk mechanism. Specifically, the conditional loss probability $\PP[Y_{r,i}=1 \mid \Psivec]$ is 
expressed as an increasing function of the latent factor
\begin{equation}
\label{eq:MaxSp}
\max\{\nu_r+\sigma_r\Psi_r,\mu_r+\sigma_r\Psi_0\},
\end{equation}
where the class-specific parameters satisfy $\nu_r,\mu_r\in\mathbb{R}$ and $\sigma_r\geq 0$.
Then, the effect in \eqref{eq:MaxSp} is mapped to the unit interval with the help of the distribution function $F_\Psi$, i.e.,
\begin{equation}
\label{eq:lossprob}
\PP[Y_{r,i}=1 \mid \Psivec]=F_\Psi\bigl(\max\{\nu_r+\sigma_r\Psi_r,\mu_r+\sigma_r\Psi_0\}\bigr).
\end{equation}
There is thus a competition between the class-specific effect, $\nu_r+\sigma_r\Psi_r$, and the global
effect, $\mu_r+\sigma_r\Psi_0$. Only the larger of the two has an impact on the occurrences of losses.
The parameters $\nu_r$ and $\mu_r$ represent the sensitivity of the conditional loss probability to the
class-specific factor $\Psi_r$ and to the global risk factor $\Psi_0$, respectively: the smaller $\mu_r$, the less
sensitive the loss indicators in category $r$ to $\Psi_0$.

Natural candidates for $F_\Psi$ are to be found among the max-stable distributions.
Max-stability ensures that the distribution of the maximum in \eqref{eq:MaxSp} stays in the same family. In this paper,
we consider the Gumbel distribution, but a similar analysis can be carried out with any other max-stable family of distributions.
Recall that the distribution function of the Gumbel (or Fisher-Tippett Extreme Value type 1) distribution
is $x\mapsto\exp\big(-\exp(-\frac{x-m}{s})\big)$ for some $m\in\mathbb{R}$ and $s>0$. We consider it here in
standardized form ($m=0$ and $s=1$) so that
$$
F_\Psi(x)=\exp\bigl(-\exp(-x)\bigr),\qquad x\in\mathbb{R}.
$$
The choice of the Gumbel distribution explains why we have chosen the latent shock to be of the form \eqref{eq:MaxSp}: the maximum in \eqref{eq:MaxSp} is again a Gumbel distributed random variable, due to the fact that the multiplicative coefficient $\sigma_r$ is equal for every element of $\bm{\Psi}$. The smaller the constants $\nu_r$ and $\mu_r$, the less sensitive the contract is to the corresponding factor.

The model we propose is related to similar constructions suggested in the literature,
but applied to different levels.
For instance, in \citet[Example 2.7]{denuit2002} it is suggested, following \citet{cossette2002}, to represent
the loss indicator $Y_{r,i}$ in terms of independent Bernoulli random variables $J_0,J_1,\ldots,J_k$
as
$$
Y_{r,i}=\min\{J_r+J_0,1\}=\max\{J_0,J_r\},\qquad i=1,\ldots,n;
$$
see also \citet{valdez2013}.
In credit risk modelling, time-to-defaults are sometimes assumed to be subject to a competing-risk mechanism \citep{giesecke2003}. Default indicators are then of the form
$$
Y_{r,i}=\mathrm{I}\big[\min\{E_r,E_0\}\leq 1\big]=\max\big\{\mathrm{I}[E_r\leq 1],\mathrm{I}[E_0\leq 1]\big\},
$$
where $E_0,E_1,\ldots,E_k$ are independent, positive random variables. The factor $E_0$ impacting all 
obligors accounts for a systematic shock
threatening the solvency of the entire portfolio.
In the model we propose, the max-factor decomposition affects the conditional loss probability and not the
loss indicators directly. Contrarily to the two models described above, where the occurrence of the
common shock ($J_0$ in the first case, or $\{ E_0\leq 1 \}$ in the second case) leads to the simultaneous occurrence of losses,
the factors $\Psi_0,\ldots,\Psi_k$ only impact the conditional loss probabilities in~\eqref{eq:lossprob}. As long as this conditional probability
stays below unity, there is still room for distinct individual default experiences. In this sense, the max-factor model appears to be more flexible.

The max-factor model can also be seen as a regime-switching construction, where the maximum drives the switch from
standard to severe conditions. Think for instance of life insurance.  The indicator $Y_{r,i}$ is now equal to 1 if individual $i$
from risk class $r$ dies during the year. Modern actuarial calculations recognize the uncertainty surrounding
one-year death probabilities. The max-factor model can account for the occurrence of pandemics increasing
the mortality of the population: $\Psi_0$ is related to the severity of the pandemics and the parameters $\mu_r$ and $\sigma_r$
modulate its consequences for the different risk categories (typically, flu pandemics can have different consequences
depending on age category). There is thus a switch in the mortality regime, from standard to high.

Compared to the classical linear specification, the maximum in \eqref{eq:MaxSp} prohibits any compensation between the global factor, $\Psi_0$, and
the category-specific factors, $\Psi_1,\ldots,\Psi_k$. Indeed, the linear combination $\mu_r+\tau_r\Psi_r+\sigma_r\Psi_0$, where
$\mu_r\in\mathbb{R}$, $\tau_r\geq 0$, $\sigma_r\geq 0$, allows for diversification between $\Psi_0$ and $\Psi_r$:
a large realization for $\Psi_0$ can be compensated by a small realization for $\Psi_r$, leaving the corresponding linear
combination unchanged.
Assume for instance that the global economy is booming, so that $\Psi_0$ is small (default probabilities being increasing
in the linear combination of risk factors). However, firms in some category $r$ may experience severe problems because
of new regulations, embargo, emerging new technologies, etc., so that $\Psi_r$ may be large. The linear combination
somewhat compensates the difficulties specific to category~$r$ with the excellent global conditions. In contrast,
the max-factor specification \eqref{eq:MaxSp} focuses on the worst factor, which is $\Psi_r$ in our example, and recognizes the particular problems
faced by the firms in category $r$. Depending on the kind of application, linear or max-factor decomposition may be
considered to represent the correlation structure of the individual loss indicators.

\section{Calibration of max-factor models}\label{sec:calibration}

\subsection{General setup}

Assume that a portfolio of risks has been observed for $n$ calendar years. Define the indicator variables 
$Y_{r,j,i}$, $r \in \{1,\ldots,k\}$, $j \in \{1,\ldots,n\}$, $i \in \{1,\ldots,m_{r,j} \}$,
where $Y_{r,j,i} = 1$ corresponds to the occurrence of losses for individual $i$ in category $r$ during 
calendar year $j$, while $m_{r,j}$ denotes the number of risks in category $r$ and calendar year $j$. In the credit risk data that we will study in Section~\ref{sec:numeric}, the categories will correspond to the rating classes. 

For fixed $r$ and $j$, the number of risks producing losses is $M_{r,j} = \sum_{i=1}^{m_{r,j}} Y_{r,j,i}$. We
assume that, within a category $r$, individual risks are exchangeable. More specifically, let
$\bm{Q}_j = (Q_{1,j},\ldots,Q_{k,j})$ be the conditional loss probabilities for calendar year $j$. Assume that
$\bm{Q}_1,\ldots,\bm{Q}_n$ are independent and identically distributed.
Given $ \bm{Q}_j$, the random variables $Y_{r,j,i}$ are independent Bernoulli random variables with respective means 
$Q_{r,j}$, so that the conditional distribution of $M_{r,j}$ is given by
\begin{equation*}
\PP[M_{r,j} = \ell \mid \bm{Q}_j] = \binom{m_{r,j}}{\ell} Q_{r,j}^{\ell} \left( 1 - Q_{r,j}\right)^{m_{r,j}-\ell},
\qquad \ell\in\{0,\ldots,m_{r,j}\}.
\end{equation*}
Conditionally on $\bm{Q}_j$, the numbers $M_{1,j},\ldots,M_{k,j}$ of risks producing losses are independent and binomially distributed. 

\subsection{Quantities of interest}\label{sec:qinterest}

We are interested in the estimation of the following quantities: 
\begin{description}
\item[Marginal loss probabilities.]
The  probability that  risk $i$ in category $r$ produces a loss during year $j$ is given by
\begin{equation}\label{eq:pi}
  \pi_r = \PP[ Y_{r,j,i} = 1 ] = \EE[ Y_{r,j,i}] = \EE[ Q_{r,j} ].
\end{equation}
\item[Joint loss probabilities.]
The probability that two different risks $i_1$ and $i_2$ in the 
same or different categories $r$ and $s$ produce losses during the same year $j$ is given by
\begin{equation}\label{eq:pijoint}
  \pi_{rs} 
  = \PP[ Y_{r,j,i_1} = 1, \, Y_{s,j,i_2} = 1 ]
  = \EE[ Y_{r,j,i_1} Y_{s,j,i_2} ] 
  = \EE[ Q_{r,j} Q_{s,j} ].
\end{equation}
\item[Intra-class higher-order loss probabilities.]
The probability that $\ell \geq 1$ risks within the same category $r$ 
produce losses during the same year $j$ is equal to
\begin{equation*}
\pi_{r}^{(\ell)} = \PP [Y_{r,j,1} = \ldots = Y_{r,j,\ell} = 1 ] = \EE [ Q_{r,j}^{\ell} ].
\end{equation*}
Clearly, $\pi_r^{(1)} = \pi_r$ and $\pi_r^{(2)} = \pi_{rr}$. 
\item[Inter-class higher-order joint loss probabilities.]
The probability that $\ell_1 \geq 1$ risks in category $r$ and $\ell_2 \geq 1$ 
risks in category $s$, where $r \neq s$, produce losses during the same year $j$ is given by
\begin{equation*}
\pi_{rs}^{(\ell_1,\ell_2)} = \PP [Y_{r,j,1} = \ldots = Y_{r,j,\ell_1} = 1, Y_{s,j,1} = \ldots = Y_{s,j,\ell_2}= 1] = \EE [Q_{r,j}^{\ell_1} Q_{s,j}^{\ell_2}]. 
\end{equation*}
Clearly, $\pi_{rs}^{(1,1)} = \pi_{rs}$. 
\end{description}
The higher-order (joint) loss probabilities are not of primary interest; they will appear in Section~\ref{sec:nonpar} where we will define nonparametric estimators for $\pi_r$ and $\pi_{rs}$. 

Dependence measures are easily expressed in terms of the probabilities defined above. For instance, the relative risk measure, or risk
ratio, used in \citet{valdez2013} in motor insurance, can be written as
$$
\frac{\PP [Y_{r,j,i_1}=1|Y_{s,j,i_2}=1]}{\PP [Y_{r,j,i_1}=1|Y_{s,j,i_2}=0]}
=\frac{\pi_{rs}(1-\pi_{s})}{(\pi_r-\pi_{rs})\pi_{s}}.
$$
Borrowed from medical studies, this quantity measures the tendency of one risk to induce another risk to produce losses. 
As pointed out in \citet{valdez2013}, the linear correlation coefficient is less suitable as a measure of association between binary random variables.
For more details, see e.g.\ \citet{denuit2005}.

\subsection{Factor models} 
\label{sec:par}

We assume a parametric model for the conditional default probabilities $\bm{Q}_j$ by setting $Q_{r,j} = Q_r (\bm{\Psi}_j ; \thetavec)$,
where $\bm{\Psi}_j = (\Psi_{1,j},\ldots,\Psi_{p,j})$ with $p<m_{r,j}$ for $j \in \{1,\ldots,n\}$ 
are independent and identically distributed latent factors with some known distribution, $\thetavec$ 
is the parameter vector, and $Q_r (\cdot \, ; \thetavec)$ are functions from $\mathbb{R}^p$ to $[0,1]$. 
To simplify the notation, we will usually omit the dependence on $\thetavec$.

Formally, for fixed $r$ and $j$, given $p$-dimensional vectors $\bm{\Psi}_j$ with $p<m_{r,j}$, $\bm{Y}_{r,j}$ follows a Bernoulli mixture model with factor vector $\bm{\Psi}_j$ if there exist functions $Q_r: \mathbb{R}^p \rightarrow [0,1]$, $r \in \{1,\ldots,k \}$, such that given $\bm{\Psi}_j = \bm{\psi}_j$, $\bm{Y}_{r,j} $ is a vector of independent Bernoulli variables with $\PP[Y_{r,j,i} = 1 \mid \bm{\Psi}_j = \bm{\psi}_j] = Q_r (\bm{\psi}_j)$ for $i \in \{1,\ldots,m_{r,j}\}$, where $\bm{\psi}_j = (\psi_{j,1},\ldots,\psi_{j,p})$. Dependence between loss indicators is essentially dependence of conditional loss probabilities on a set of factors. 

As described in Section~\ref{sec:maxfactor}, our focus is on a Gumbel max-factor model. For comparison, we consider factor models based on the normal distribution and a Gumbel one-factor model as well.
\begin{description}
\item[Model (1a)]
The one-factor Probit-Normal specification assumes that 
for every year $j$
\begin{equation*}
Q_{r} \left( \Psi_j \right) = \Phi(\mu_{r} + \sigma_{r} \Psi_j ), \qquad \sigma_{r} >0, \, r \in \{1,\ldots,k\},
\end{equation*}
where $\Phi$ denotes the standard Normal distribution function and $\Psi_1,\ldots,\Psi_k$
are independent with common distribution function $\Phi$ for $j \in \{1,\ldots,n\}$. 
The model parameters are $\thetavec = (\mu_1,\ldots,\mu_k,\sigma_1,\ldots,\sigma_k)$. 
This classical model has been applied in \citet{frey2003} to the same dataset appearing in Section~\ref{sec:numeric}.
\item[Model (2a)]
A direct extension of model (1a) is
\begin{equation*}
Q_{r} \left( \bm{\Psi}_j \right)
= \Phi \left( \mu_{r} + \tau_{r} \Psi_{r,j} + \sigma_{r} \Psi_{0,j} \right), \qquad \sigma_{r},\tau_{r} >0, \, r \in \{1,\ldots,k\},
\end{equation*}
where the $k + 1$ components $\Psi_{0,j}, \Psi_{1,j}, \ldots,\Psi_{k,j}$ of $\bm{\Psi}_j$ are independent with common distribution function $\Phi$. The model parameters are $\thetavec = (\mu_1,\ldots,\mu_k,\tau_1,\ldots,\tau_k,\sigma_1,\ldots,\sigma_k)$. 
If $\tau_r \rightarrow 0$ for every $r$, we retrieve model (1a).
\item[Model (1b)] 
The Gumbel one-factor can be defined as
\begin{equation*}
Q_{r} \left( \Psi_j \right) = F_{\Psi} (\mu_{r} + \sigma_{r} \Psi_j),  \qquad \sigma_{r} >0, \, r \in \{1,\ldots,k\},
\end{equation*}
where the factors $\Psi_1,\ldots,\Psi_k$ are independent random variables with common distribution function $F_{\Psi} (x) =  \exp \left( - \exp(-x) \right)$. The vector of model parameters is $\thetavec = (\mu_1,\ldots,\mu_k,\sigma_1,\ldots,\sigma_k)$. 
\item[Model (2b)]
For the Gumbel max-factor model, we take 
\begin{equation*}
Q_{r} \left( \bm{\Psi}_j \right)
= F_{\Psi} \left( \max{\{\nu_{r} + \sigma_{r} \Psi_{r,j}, \mu_{r} + \sigma_{r} \Psi_{0,j} \}} \right), \qquad \sigma_{r}  >0,
\end{equation*}
where $\Psi_{0,j}, \Psi_{1,j}, \ldots,\Psi_{k,j}$ are independent with common distribution function $F_\Psi$. The model parameters are $\thetavec = (\nu_1,\ldots,\nu_k,\mu_1,\ldots,\mu_k,\sigma_1,\ldots,\sigma_k)$. 
If $\nu_r \rightarrow -\infty$ for every $r$, then we are back at model (1b).
\end{description}
Models (1a)-(1b) involve a single factor but differ in the right tails of the conditional loss probabilities $Q_{r} \left( \bm{\Psi}_j \right)$:
the probability that these conditional probabilities exceed high thresholds is typically larger under the Gumbel specification
compared to the Gaussian one. Considering models (2a)-(2b), a global effect $\Psi_{0,j}$ is now combined with category-specific
effects $\Psi_{r,j}$. This gives more flexibility as conditional loss probabilities now become dependent, sharing the common random
effect $\Psi_{0,j}$. It is worth mentioning that the interpretation of the parameters is different under models (2a) and (2b).
In model (2a), the coefficients $\tau_r$ and $\sigma_r$ multiplying the random effects measure the sensitivity of the individual
risks in category $r$ to $\Psi_{r,j}$ and $\Psi_{0,j}$, respectively, whereas these sensitivities are measured by the additive parameters
$\nu_r$ and $\mu_r$ in model (2b).

As described in Section~\ref{sec:qinterest}, we focus on the marginal loss probabilities $\pi_r$ and the joint loss probabilities $\pi_{rs}$. If $F_{\bm{\Psi}}$ is the distribution function of a generic risk factor $\bm{\Psi}$, then these loss probabilities are obtained directly from \eqref{eq:pi} and \eqref{eq:pijoint} by
\begin{align}
 \pi_r & = \EE [Q_r(\bm{\Psi})] = \int Q_r(\bm{\psi}) \, \textrm{d}F_{\bm{\Psi}} (\bm{\psi}), \label{pir} \\
  \pi_{rs} & = \EE [Q_r(\bm{\Psi})\,Q_s(\bm{\Psi})] = \int Q_r(\bm{\psi}) \, \, Q_s(\bm{\psi}) \, \textrm{d}F_{\bm{\Psi}} (\bm{\psi}). \label{pirs}
\end{align}

\subsection{Likelihood} 
\label{sec:likelihood}
Let $F_{\bm{\Psi}}$ denote again the distribution function of a generic risk factor $\bm{\Psi}$. For category $r$ and year $j$, the unconditional distribution of the number of risks producing losses is given by 
\begin{equation*}
\PP[M_{r,j} = \ell_{r,j}] = \binom{m_{r,j}}{\ell_{r,j}} \int Q_r \left(\bm{\psi}_j\right)^{\ell_{r,j}} 
\big( 1-Q_r(\bm{\psi}_j) \big)^{m_{r,j}-\ell_{r,j}} \, \textrm{d} F_{\bm{\Psi}} (\bm{\psi}_j).
\end{equation*}
We write $\bm{M}_j = (M_{1,j},\ldots,M_{k,j})$ and $\ellvec_j = (\ell_{1,j},\ldots,\ell_{k,j})$ for $j=1,\ldots,n$.
Notice that since the loss indicators are independent given the vectors $\bm{\Psi}_j$, we can write
\begin{equation}\label{eq:nrdef2}
\PP[ \bm{M}_j = \ellvec_j \mid \bm{\Psi}_j = \bm{\psi}_j] = \prod_{r=1}^k \binom{m_{r,j}}{\ell_{r,j}} 
 Q_r(\bm{\psi}_j)^{\ell_{r,j} } \big( 1 - Q_r(\bm{\psi}_j) \big)^{m_{r,j} - \ell_{r,j}}.
\end{equation}
For every year we have expression \eqref{eq:nrdef2} and the log-likelihood takes the form
\begin{equation*}
L_n(\thetavec; \bm{M}_1,\ldots,\bm{M}_n) = \sum_{j=1}^n \sum_{r=1}^k \log{ \binom{m_{r,j}}{M_{r,j}} } + \sum_{j=1}^n \log{I_j},
\end{equation*}
where 
\begin{equation*}
I_j = \int \prod_{r=1}^k  Q_r(\bm{\psi}_j)^{M_{r,j}} \big( 1 - Q_r(\bm{\psi}_j) \big)^{m_{r,j} - M_{r,j}}  \, \textrm{d}F_{\bm{\Psi}} (\bm{\psi}_j).
\end{equation*}
For the one-factor models (1a) and (1b), we find it convenient to make the substitution $q = F_{\Psi} (\psi_j)$ and to evaluate $I_j$ as
\begin{equation*}
I_j = \int_0^1 \exp \left( \sum_{r=1}^k  M_{r,j} \log{\left( Q_r(F_{\Psi}^{-1} (q)) \right)} + (m_{r,j} - M_{r,j}) \log{ \left( 1 - Q_r(F_{\Psi}^{-1} (q) ) \right)} \right)  \, \textrm{d}q.
\end{equation*}
For models (2a) and (2b), we can make the substitutions $q_l = F_{\Psi} (\psi_{j,l})$ for $l \in \{0,\ldots,k\}$ since $\bm{\psi}_j = (\psi_{j,0}, \ldots, \psi_{j,k})$. Then, we can write the likelihood as 
\begin{equation}\label{eq:lik1}
I_j =  \int_{[0,1]^{k+1}} \left( \prod_{r=1}^k f_{r,j} (q_r,q_0) \right) \, \textrm{d}q_0 \cdots \textrm{d}q_k , 
\end{equation}
where 
\begin{equation}\label{eq:lik2}
f_{r,j} (q_r,q_0) = \Bigl(Q_r \left( F_{\Psi}^{-1} (q_r), F_{\Psi}^{-1} (q_0) \right)\Bigr)^{M_{r,j}} \, 
\Bigl( 1 - Q_r \left( F_{\Psi}^{-1} (q_r), F_{\Psi}^{-1} (q_0) \right) \Bigr)^{m_{r,j} - M_{r,j}}.
\end{equation}
Each likelihood term involves high-dimensional numerical integration over a complicated function. Especially for model (2b), when the integrand is a product of maxima, a nondifferentiable function, this is a computational burden. Fortunately, we can simplify the likelihood to a sum of lower-dimensional integrals over smoother functions thanks to the following result.
\begin{Lemma}\label{lemma1}
Define $I_j$ and $f_{r,j}$ for $r=1,\ldots,k$ and $j=1,\ldots,n$ as in \eqref{eq:lik1} and \eqref{eq:lik2}, where $Q_r$ is the function corresponding to the Gumbel max-factor model, Model (2b). Define
\begin{align*}
g_r (q) & = \exp \left\{ \log(q) \exp \left( \frac{\nu_r - \mu_r}{\sigma_r} \right) \right\}, \\
h_{r,j} (q ; \mu_r) & = F_{\Psi} (\mu_r - \sigma_r \log \left( - \log (q) \right))^{M_{r,j}}  \times 
\left( 1 - F_{\Psi} (\mu_r - \sigma_r \log \left( - \log (q) \right)  \right)^{m_{r,j} - M_{r,j}}.
\end{align*}
Let $R = \{1,\ldots,k\}$ and let $\mathcal{P}(R)$ denote the power set of $R$. Then
\begin{equation}\label{eq:finallik}
I_j = \sum_{I \in \mathcal{P}(R)} \int_0^1 \left( \prod_{r \in R \setminus I} g_r (q_0) \, h_{r,j} (q_0 ; \mu_r) \right) \left(\, \prod_{r \in I} \int_{g_r (q_0)}^1 h_{r,j} (q_r ; \nu_r) \, \mathrm{d} q_r \right) \, \mathrm{d} q_0.
\end{equation}
\end{Lemma}
The proof of this result is provided in Appendix~\ref{sec:proofs}.

The parameter vector $\thetavec$ is estimated by maximizing the log-likelihood $L_n (\bm{\theta})$.
After estimating $\thetavec$, the implied marginal and joint loss probabilities, \eqref{eq:pi} and \eqref{eq:pijoint}, are obtained by plugging in the estimator of $\thetavec$ in expressions \eqref{pir} and \eqref{pirs}, yielding $\widehat{\pi}_r$ and $\widehat{\pi}_{rs}$, respectively.

\section{Nonparametric estimation}\label{sec:nonpar}
Nonparametric estimators can be useful as a benchmark for model-based estimators, especially in the case of model uncertainty. Usually, the nonparametric estimators presented in Section~\ref{noweight} are used, see for example \citet{frey2003}. However, since the numbers $m_{r,j}$ may vary strongly over the years, more accurate nonparametric estimators are obtained by assigning more weight to those years for which there is more information (Section~\ref{weight}).

\subsection{Preliminary estimators}\label{noweight}

Define the observed proportions of risks producing losses as  
\begin{equation*}
\widehat{Q}_{r,j} = M_{r,j} / m_{r,j}, \qquad \text{ for } r \in \{1,\ldots,k\}, \, j \in \{1,\ldots,n\}. 
\end{equation*}
For $r \neq s$, define the estimators
\begin{align}
\widehat{\pi}^{(\ell)}_r & = \frac{1}{n} \sum_{j=1}^n \frac{M_{r,j} (M_{r,j} - 1) \cdots (M_{r,j} - \ell + 1)}{m_{r,j}(m_{r,j} - 1) \cdots (m_{r,j} - \ell + 1)}, 
\quad \text{ for }\ell < m_{r,j}, \label{Prel1}\\
\widehat{\pi}^{(\ell_1,\ell_2)}_{rs} & = \frac{1}{n} \sum_{j=1}^n \frac{M_{r,j} (M_{r,j} - 1) \cdots 
(M_{r,j} - \ell_1 + 1)}{m_{r,j}(m_{r,j} - 1) \cdots (m_{r,j} - \ell_1 + 1)} \, \frac{M_{s,j} (M_{s,j} - 1) \cdots (M_{s,j} - \ell_2 + 1)}{m_{s,j}(m_{s,j} - 1) \cdots (m_{s,j} - \ell_2 + 1)},\label{Prel2}  
\end{align}
for $\ell_1 < m_{r,j}$ and $\ell_2 < m_{s,j}$.
To see that \eqref{Prel1} and \eqref{Prel2} are unbiased estimators, recall that if the random variable $M$ is binomially distributed
with $n$ trials and success probability $p$, then we have for $\ell \in \{1, \ldots, n\}$ that
\begin{equation}\label{eq:binmoment}
  \EE[ M (M - 1) \cdots (M - \ell+1) ] = n(n-1)\cdots(n-\ell+1) \, p^\ell.
\end{equation}
Conditionally on $\bm{Q}_j $,  the random variable $M_{r,j}$ follows the binomial distribution with $m_{r,j}$ trials and success probability $Q_{r,j}$. Hence, for $\ell < m_{r,j}$, 
\begin{eqnarray*}
\EE [ \widehat{\pi}_r^{(\ell)}]& =& \frac{1}{n} \sum_{j=1}^n   \EE \left[ \EE \left[ 
    \frac%
      { M_{r,j} (M_{r,j} - 1) \cdots (M_{r,j} - \ell + 1) }%
      { m_{r,j} (m_{r,j} - 1) \cdots (m_{r,j} - \ell + 1) } \, \middle| \, \bm{Q}_j \right]
  \right] \\
 & = &\frac{1}{n} \sum_{j=1}^n \EE [ Q_{r,j}^\ell ] = \pi_r^{(\ell)},
\end{eqnarray*}
and similarly, $\EE [\widehat{\pi}_{rs}^{(\ell_1,\ell_2)}] = \pi_{rs}^{(\ell_1,\ell_2)}$. 

\subsection{Weighted estimators}\label{weight}

For the marginal loss probabilities $\pi_r$, consider estimators of the form
\[
  \widetilde{\pi}_r (\bm{w}_r) = \sum_{j=1}^n w_{r,j} \widehat{Q}_{r,j}, \qquad r \in \{1,\ldots,k\}, 
\]
where $\widehat{Q}_{r,1},\ldots,\widehat{Q}_{r,n}$ have a common expectation $\EE[ \widehat{Q}_{r,j} ] = \pi_r$ and possibly different variances $\mathrm{Var} [ \widehat{Q}_{r,j} ] = \sigma_{r,j}^2$, 
and where the weight vector $\bm{w}_r = (w_{r,1}, \ldots, w_{r,n})$ has nonnegative entries. 
We seek optimal weights, in the sense that we minimize the mean squared error of $\widetilde{\pi}_r (\bm{w}_r)$ as a function of $\bm{w}_r$, leading to weights $w_{r,1,\opt},\ldots,w_{r,n,\opt}$.
The same recipe can be followed for the joint loss probabilities $\pi_{rr}$ and $\pi_{rs}$. 

\begin{Theorem}\label{thmopt0}
The estimators $(\pi_{r,\opt},\pi_{rr,\opt},\pi_{rs,\opt})$
for $r,s, \in \{1,\ldots,k\}$ and $r \neq s$ that minimize the mean squared error of $(\widetilde{\pi}_r (\bm{w}_r), \widetilde{\pi}_{rr} (\bm{w}_{r}), \widetilde{\pi}_{rs} (\bm{w}_r))$ are
\begin{enumerate}
\item $\begin{aligned}[t]
  \pi_{r,\opt}  & = \sum_{j=1}^n w_{r,j,\opt} \frac{M_{r,j}}{m_{r,j}}, \qquad  w_{r,j,\opt}  = \frac{\sigma_{r,j}^{-2}}{\pi_r^{-2} + \sum_{t=1}^n \sigma_{r,t}^{-2}}, \\ 
   \sigma_{r,j}^2 & =  \frac{\pi_r}{m_{r,j}} + \left( 1 - \frac{1}{m_{r,j}} \right) \pi_{rr} - \pi^2_r;
\end{aligned}$
\item $\begin{aligned}[t]
  \pi_{rr,\opt}  &= \sum_{j=1}^n w_{rr,j,\opt} \frac{M_{r,j} (M_{r,j}-1)}{m_{r,j} (m_{r,j}-1)}, \qquad
   w_{rr,j,\opt}  = \frac{\sigma_{rr,j}^{-2}}{\pi_r^{-2} + \sum_{t=1}^n \sigma_{rr,t}^{-2}}, \\
   \sigma_{rr,j}^2 & =  m_{r,j} (m_{r,j} - 1) \bigl[ (2 - m_{r,j} (m_{r,j} - 1) \pi_{rr}) \pi_{rr} \\ 
& \qquad  + 4 (m_{r,j} - 2) \pi^{(3)}_{r} + (m_{r,j} -2) (m_{r,j} -3) \pi^{(4)}_{r} \bigr];
\end{aligned}$
\item $\begin{aligned}[t]
  \pi_{rs,\opt} &= \sum_{j=1}^n w_{rs,j,\opt} \frac{M_{r,j} M_{s,j}}{m_{r,j} m_{s,j}}, \qquad
w_{rs,j,\opt} = \frac{\sigma_{rs,j}^{-2}}{\pi_r^{-2} + \sum_{t=1}^n \sigma_{rs,t}^{-2}}, \\
\sigma_{rs,j}^2 &  =  m_{r,j}^{-1} m_{s,j}^{-1} \bigl[ (1 - m_{r,j} m_{s,j} \pi_{rs}) \pi_{rs} + (m_{s,j} - 1) \pi^{(1,2)}_{rs} + (m_{r,j} - 1) \pi^{(2,1)}_{rs} \\ 
& \qquad + (1 - m_{s,j} - m_{r,j} + m_{r,j} m_{s,j} \pi^{(2,2)}_{rs} ) \bigr] .
\end{aligned}$
\end{enumerate}
The variances of these estimators are given by
\begin{align}
\label{eq:asvar}
\mathrm{Var}  [\pi_{r,\opt}] &= \frac{1}{ \sum_{j=1}^n \sigma_{r,j}^{-2}}, &
\mathrm{Var}  [\pi_{rr,\opt}] &= \frac{1}{\sum_{j=1}^n \sigma_{rr,j}^{-2}}, &
\mathrm{Var} [\pi_{rs,\opt} ] &= \frac{1}{ \sum_{j=1}^n \sigma_{rs,j}^{-2}}. 
\end{align}
\end{Theorem}
The proof of this result is provided in Appendix~\ref{sec:proofs}.
As the quantities $\sigma_{r,j}$, $\sigma_{rr,j}$, and $\sigma_{rs,j}$ depend on the unknown quantities $\pi_r^{(\ell)}$ and $\pi_{rs}^{(\ell_1,\ell_2)}$, we replace these with their preliminary estimators $\widehat{\pi}_r^{(\ell)}$ and $\widehat{\pi}_{rs}^{(\ell_1,\ell_2)}$ from \eqref{Prel1} and \eqref{Prel2} respectively.

\begin{Definition}\label{thmopt}
Let $\widehat{\pi}_{r}^{(\ell)}$ and $\widehat{\pi}_{rs}^{(\ell_1,\ell_2)}$ be defined as in \eqref{Prel1} and \eqref{Prel2}. 
The estimators \\ $(\widehat{\pi}_{r,\opt},\widehat{\pi}_{rr,\opt},\widehat{\pi}_{rs,\opt})$ of $(\pi_r, \pi_{rr}, \pi_{rs})$ 
for $r,s, \in \{1,\ldots,k\}$ and $r \neq s$ are defined as
\begin{enumerate}
\item $\begin{aligned}[t]
  \widehat{\pi}_{r,\opt} & = \sum_{j=1}^n \widehat{w}_{r,j,\opt} \frac{M_{r,j}}{m_{r,j}}, \qquad  \widehat{w}_{r,j,\opt}  = \frac{\widehat{\sigma}_{r,j}^{-2}}{\widehat{\pi}_r^{-2} + \sum_{t=1}^n \widehat{\sigma}_{r,t}^{-2}}, \\ 
    \widehat{\sigma}_{r,j}^2 & =  \frac{\widehat{\pi}_r}{m_{r,j}} + \left( 1 - \frac{1}{m_{r,j}} \right) \widehat{\pi}_{rr} - \widehat{\pi}^2_r;
\end{aligned}$
\item $\begin{aligned}[t]
  \widehat{\pi}_{rr,\opt}&= \sum_{j=1}^n \widehat{w}_{rr,j,\opt} \frac{M_{r,j} (M_{r,j}-1)}{m_{r,j} (m_{r,j}-1)}, \qquad
    \widehat{w}_{rr,j,\opt}  = \frac{\widehat{\sigma}_{rr,j}^{-2}}{\widehat{\pi}_r^{-2} + \sum_{t=1}^n \widehat{\sigma}_{rr,t}^{-2}}, \\
    \widehat{\sigma}_{rr,j}^2 & =  m_{r,j} (m_{r,j} - 1) \bigl[ (2 - m_{r,j} (m_{r,j} - 1) \widehat{\pi}_{rr}) \widehat{\pi}_{rr} \\ 
& \qquad  + 4 (m_{r,j} - 2) \widehat{\pi}^{(3)}_{r} + (m_{r,j} -2) (m_{r,j} -3) \widehat{\pi}^{(4)}_{r} \bigr];
\end{aligned}$
\item $\begin{aligned}[t]
  \widehat{\pi}_{rs,\opt}&= \sum_{j=1}^n \widehat{w}_{rs,j,\opt} \frac{M_{r,j} M_{s,j}}{m_{r,j} m_{s,j}}, \qquad
  \widehat{w}_{rs,j,\opt} = \frac{\widehat{\sigma}_{rs,j}^{-2}}{\widehat{\pi}_r^{-2} + \sum_{t=1}^n \widehat{\sigma}_{rs,t}^{-2}}, \\
    \widehat{\sigma}_{rs,j}^2 &  =  m_{r,j}^{-1} m_{s,j}^{-1} \bigl[ (1 - m_{r,j} m_{s,j} \widehat{\pi}_{rs}) \widehat{\pi}_{rs} + (m_{s,j} - 1) \widehat{\pi}^{(1,2)}_{rs} + (m_{r,j} - 1) \widehat{\pi}^{(2,1)}_{rs} \\ 
& \qquad + (1 - m_{s,j} - m_{r,j} + m_{r,j} m_{s,j} \widehat{\pi}^{(2,2)}_{rs} ) \bigr] .
\end{aligned}$
\end{enumerate}
\end{Definition}

Approximate standard errors of these estimators can be obtained by plugging in the estimators of $\sigma_{r,j}^2$, $\sigma_{rr,j}^2$ and  $\sigma_{rs,j}^2$ into \eqref{eq:asvar}. 

A simulation study illustrating the performance of these estimators is provided in Appendix~\ref{sec:simstudy}. We found that the relative root mean squared error (RRMSE) of the estimators of $\pi_r$ and $\pi_{rs}$ is significantly lower for the weighted estimators than for the preliminary estimators. Moreover, the weighted nonparametric estimators have low RRMSE in comparison with the maximum likelihood estimators of Section~\ref{sec:likelihood}, especially when the parametric model is misspecified, that is, when we maximize the likelihood of the parameters of a factor model that is different from the true underlying model.

\section{Application to credit risk}\label{sec:numeric}

\subsection{Credit risk data}
\label{sec:data}

We study one-year default rates for groups of obligors formed into static pools (cohorts). The default rates are taken from Table~13
in \cite{sp2001}, where the period of study is 1981--2000. The total data comprises around 9200 obligors rated as of January 1st, 1981, or first rated between that date and December 31st, 1999.
A company is considered defaulted on the date when it is unable to fulfill a payment or any other financial obligation for the first time.
Companies are given credit ratings ranging from AAA to CCC.
We consider here the ratings BB, B, and CCC which form the group ``speculative grade''. The starting year, 1981, does not include companies that defaulted in that year. 
Since it contains zero defaults by construction, we removed that year from our study.

Few obligors default early in their rating history. If default rates are obtained by dividing the number of defaults by all outstanding ratings, then consequently the default rates will be comparatively low during periods of high rating activity. To avoid any misleading results, the data is presented for cohorts called static pools. A static pool is formed on the first day of each year, and includes all companies in the study. The pools are called static because their membership remains constant over time. The obligors are followed from year to year within each pool. The ratings of the first and last days of each year are compared. 
Companies that default (D) or whose ratings have been withdrawn (N.R., not rated) are excluded from subsequent pools. For instance, we start with all companies that had outstanding non-defaulted ratings on January 1st, 1981. The 1982 static pool consisted of all companies that survived 1981 plus all companies that were first rated in 1981. In the scope of our time period, 9169 first-time rated organizations were added to the static pools, 746 companies defaulted and 3118 companies were excluded due to a N.R. rating. A company usually obtains a N.R. rating due to paid-off debt, a result of mergers and acquisitions, or a lack of cooperation with the rating agency.
Figure \ref{fig:tseries} shows the total number of firms per rating class, the number of defaulted firms per rating class and the proportion of defaults per rating class.

\begin{figure}[p]
\centering
\subfloat{\includegraphics[width=0.75\textwidth]{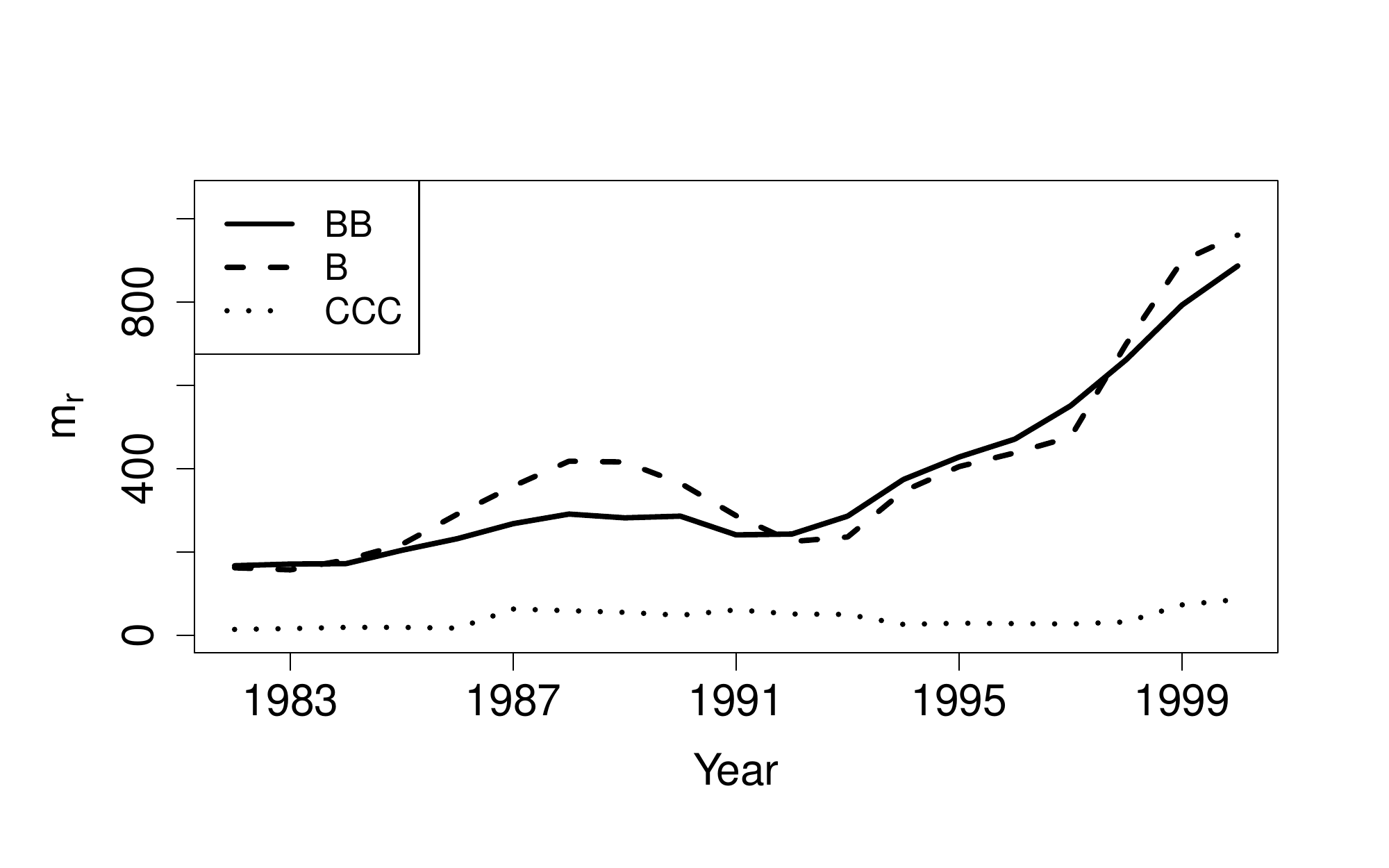}} \\ \vspace{-2cm}
\subfloat{\includegraphics[width=0.75\textwidth]{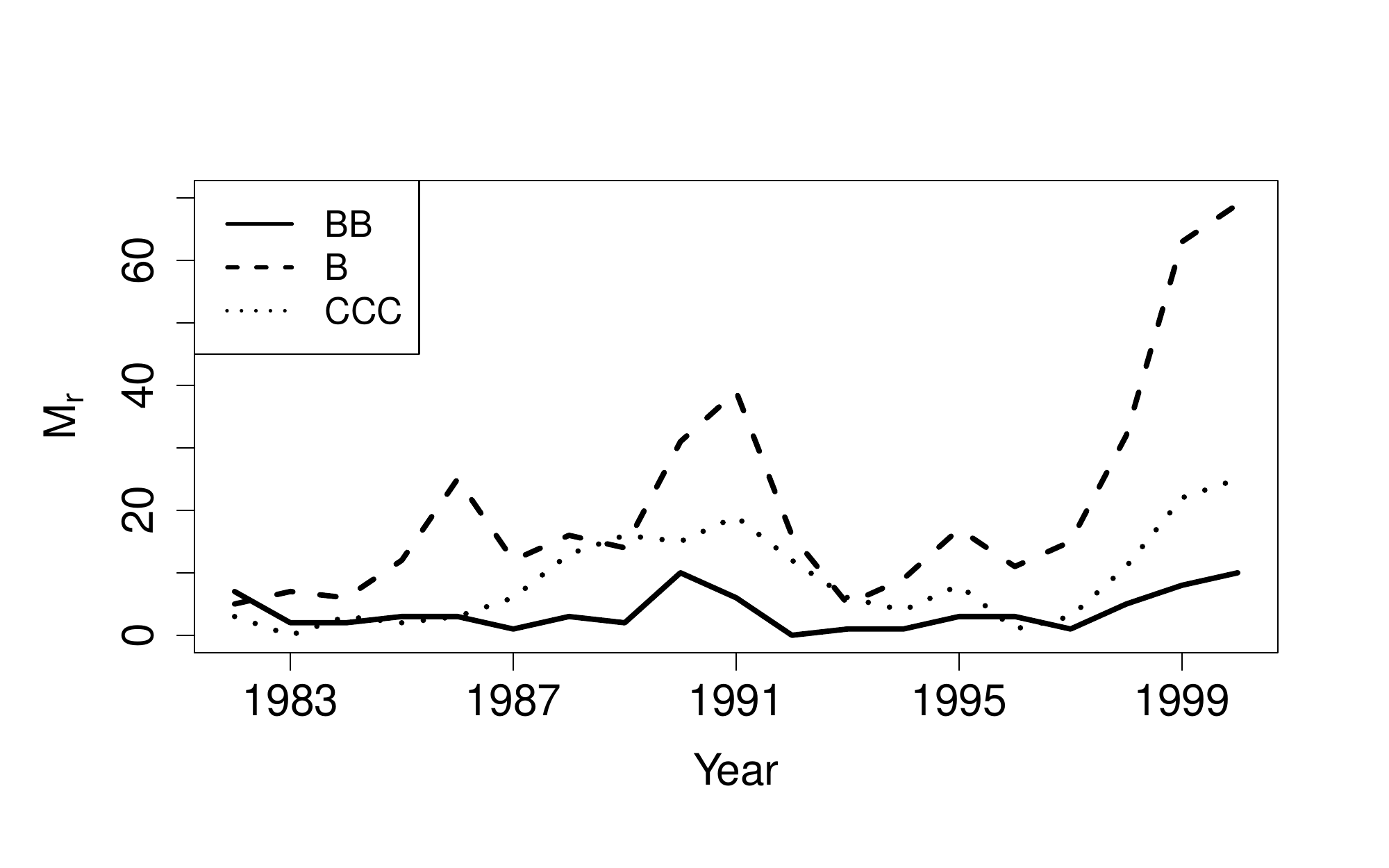}} \\ \vspace{-2cm}
\subfloat{\includegraphics[width=0.75\textwidth]{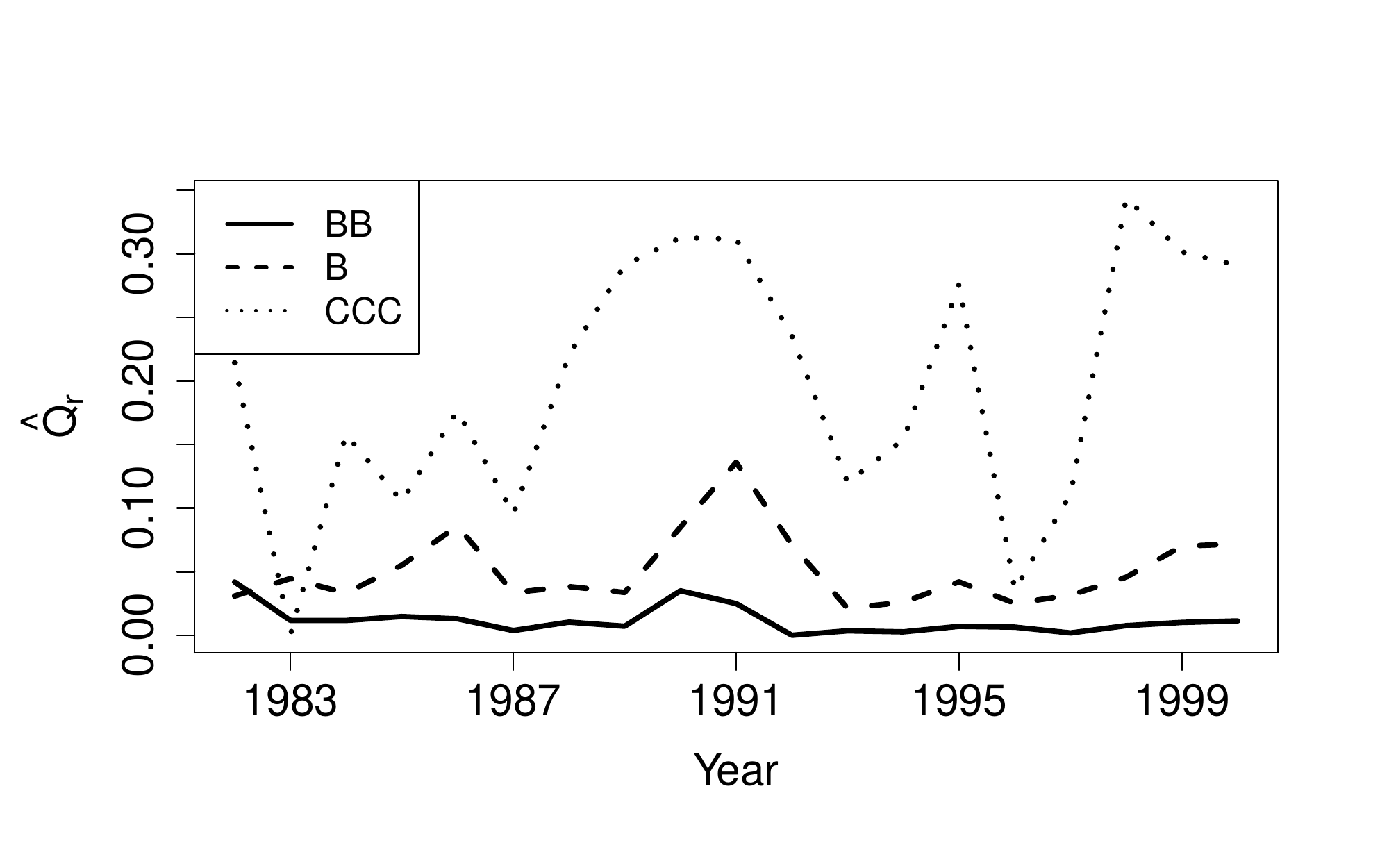}} 
\caption{\citet{sp2001} data: the total number of firms per rating class $m_r$ (top), the number of defaulted firms per rating class $M_r$ (middle) and the proportion of defaults per rating class $\widehat{Q}_r = M_{r} / m_{r}$ for the years 1982--2000, where $r \in \{ \text{BB}, \text{B}, \text{CCC}\}$.}
\label{fig:tseries}
\end{figure}

\subsection{Nonparametric estimation}

Table~\ref{resnonpar} displays  the estimators $\widehat{\pi}_{r,\opt}$ and $\widehat{\pi}_{rs,\opt}$ obtained from Definition~\ref{thmopt}, where $r,s \in \{\text{BB}, \text{B}, \text{CCC}\}$. The standard errors are in parentheses. These values serve as benchmarks to evaluate the accuracy of the parametric factor models fitted in the next section.

\begin{table}[ht]
\centering
\begin{tabular}{lccc}
\toprule
 & BB & B & CCC   \\
 \midrule
$\widehat{\pi}_{r,\opt}$ & 0.0107  (0.0024) & 0.0511 (0.0064) & 0.2069  (0.0225)  \\
\midrule
$\widehat{\pi}_{rs,\opt} \times 1000$ & BB & B & CCC   \\
\midrule
BB &  0.151 (0.081)  & 0.649 (0.206) & 2.438 (0.682) \\
B &  0.649 (0.206) & 3.075 (0.935) & 11.64 (2.438)\\
CCC &  2.438 (0.692) & 11.64 (2.438)  & 49.02 (8.887) \\
\bottomrule
\end{tabular}
  \caption{\citet{sp2001} data: Nonparametric estimators of the marginal default probabilities $\pi_r$ and the joint default probabilities $\pi_{rs}$.}
  \label{resnonpar}
\end{table}

\subsection{Parametric factor models}
We estimated the parameters of the parametric factor models (1a), (2a), (1b), and (2b).
Table \ref{AIC} shows the AIC and BIC for these models. Note that the numbers of parameters for models (2a) and (2b) is
reduced:
\begin{itemize}
\item for model (2a), we find $\tau_{B} \rightarrow 0$, i.e. class B does not require a specific random effect and is influenced by the
global effect, only;
\item for model (2b), we find $\nu_{B}, \nu_{CCC} \rightarrow -\infty$, i.e. classes B and CCC do not require specific random effects and are influenced by the
global effect, only.
\end{itemize}
In terms of both AIC and BIC, the one-factor models (1a) and (1b) perform better than a multi-factor Normal model (model 2a). However, the Gumbel max-factor model (2b) appears to be the best alternative. Between models (1b) and (2b) we can also perform a likelihood ratio test, since model (1b) is a submodel of model (2b) with $\nu_{BB} \rightarrow - \infty$. The value of the likelihood ratio test statistic is equal to $2.76$. The hypothesis for $\nu_{BB} = - \infty$ concerns a value at the boundary of the parameter space. The asymptotic null distribution of the likelihood ratio test statistic $2 \log( L_n)$ is a mixture of two chi-squared distributions; see \citet{self1987}. We reject the null hypothesis of the one-factor model at a significance level of $\alpha = 0.05$, corresponding to a critical value of $1.92$. 

Table~\ref{res1} shows estimates and standard errors for the parameters of model~(2b), together with implied estimates of the marginal default probabilites $\pi_r$ and the joint default probabilities $\pi_{rs}$, obtained using expressions \eqref{pir} and \eqref{pirs}.
Both $\widehat{\pi}_r$ and $\widehat{\pi}_{rs}$ match the nonparametric ones reasonably well.

A visual test is presented in the form of prediction intervals for the numbers of defaults, $M_{r,j}$. We simulate 5000 realizations of $Q_{BB},Q_B,Q_{CCC}$, accounting for the correlation structure, i.e., we simulate 5000 realizations of model (2b) using the parameter values that we obtained for the corresponding model. Using $Q_{BB},Q_B,Q_{CCC}$ we simulate 5000 realizations of $M_{r,j}$, where $m_{r,j}$ is given by the credit risk data. Finally, we calculate the prediction intervals, obtained by isolating the 4500 central realizations. The results are presented in Figure~\ref{fig:pred1}. 
The observed number of defaults generally stays within the prediction intervals; departures are in line with the 90\%
confidence level. These intervals provide the risk manager with useful ranges for the number of defaults.

It is also interesting to compare the distribution function of the conditional default probabilities $Q_r(\Psivec)$ for models
(1a)--(1b) to models (2a)--(2b). Figure \ref{fig2} shows the survival functions of $Q_r (\Psivec)$ for $r \in \{\text{BB}, \text{B}, \text{CCC} \}$.
The fatness of the right tail of $Q_r(\Psivec)$ greatly distinguishes the Gumbel models from the Normal ones, even if
all distribution functions agree to a large extent around the mean value. The impact of replacing the traditional normally distributed latent factor
with a Gumbel one is clearly visible for high quantiles.
\begin{table}[t]
\centering
\begin{tabular}{lccccc}
\toprule
Model & \# parameters & $- \log L_n$ & AIC & BIC  \\
\midrule
(1a) & 6 & 154.707 & 321.41 & 316.05 \\
(2a) & 8 & 154.445 & 324.89 & 317.74 \\
(1b) & 6 & 154.517 & 321.03 & 315.67 \\
(2b) & 7 & 153.138 & 320.28 & 314.02 \\
\bottomrule
\end{tabular}
  \caption{\citet{sp2001} data: overview of the number of parameters, the negative log-likelihood, AIC, and BIC for the four parametric models.}
  \label{AIC}
\end{table}

\begin{table}[ht]
\centering
\begin{tabular}{lcccc}
\toprule
\text{ } & BB & B & CCC \\
\midrule
$\mu_r$ & $-1.66$ $(0.07)$ & $-1.18$ $(0.04)$ & $-0.54$ $(0.07)$  \\
$\nu_r$ & $-1.73$ $(0.11)$ & --- & --- & \\
$\sigma_r$ & $0.112$ $(0.033)$ & $0.124$ $(0.029)$& $0.162$ $(0.053)$  \\
\midrule
$\widehat{\pi}_r$ & $0.0109$ & $0.0520$ & $0.2120$ \\
\midrule
$\widehat{\pi}_{rs} \times 1000$ & BB & B & CCC  \\
\midrule
BB & $0.215$ & $0.781$ & $2.795$  \\
B &  $0.781$ & $3.512$ & $12.96$  \\
CCC & $2.795$ & $12.96$ & $49.73$  \\
\bottomrule
\end{tabular}
  \caption{\citet{sp2001} data: maximum likelihood parameter estimates and standard errors for the Gumbel max-factor model (2b),
together with the implied estimates of default probabilities.}
  \label{res1}
\end{table}

\begin{figure}[ht]
\centering
\subfloat{\includegraphics[width=0.7\textwidth]{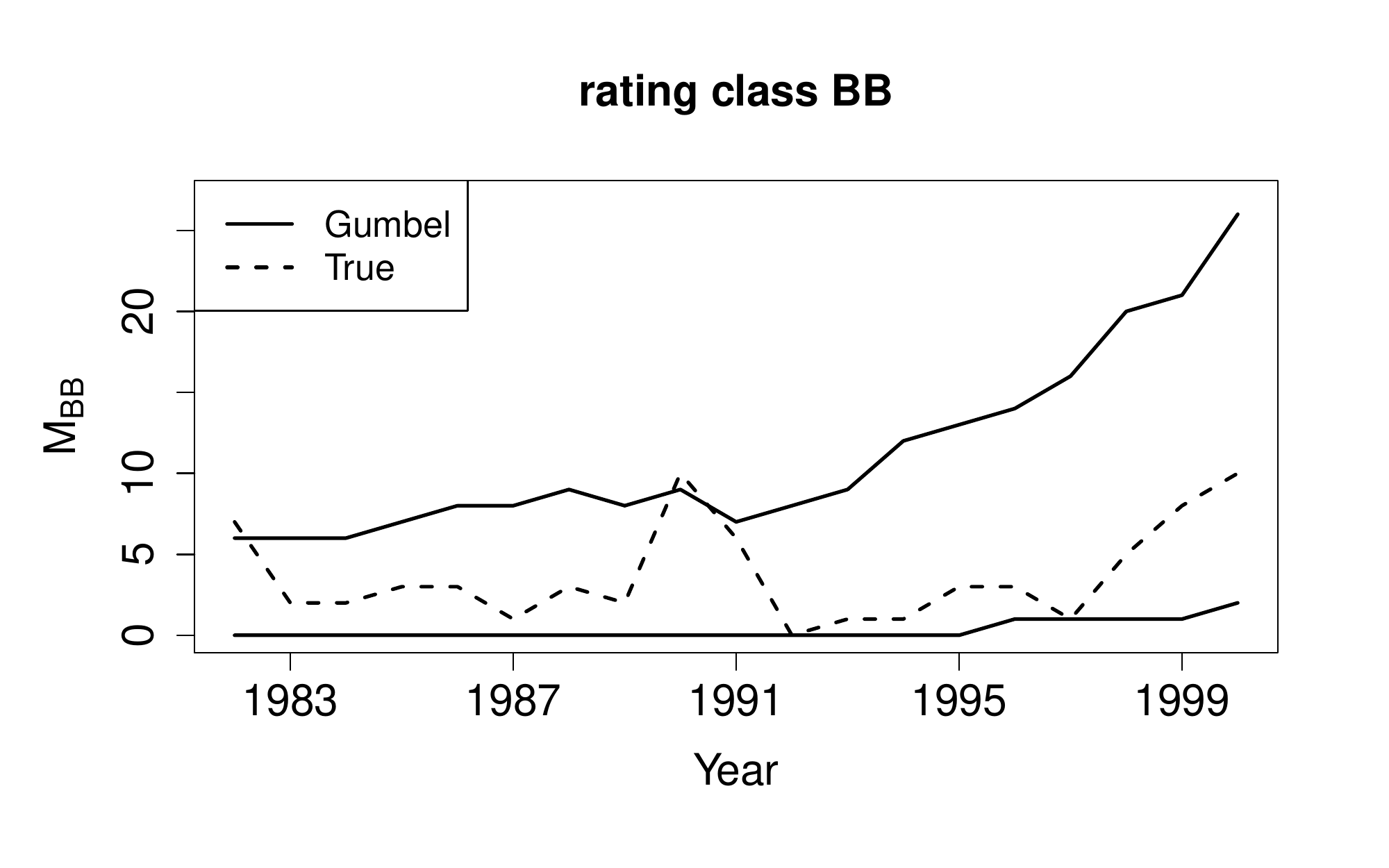}} \\ \vspace{-1cm}
\subfloat{\includegraphics[width=0.7\textwidth]{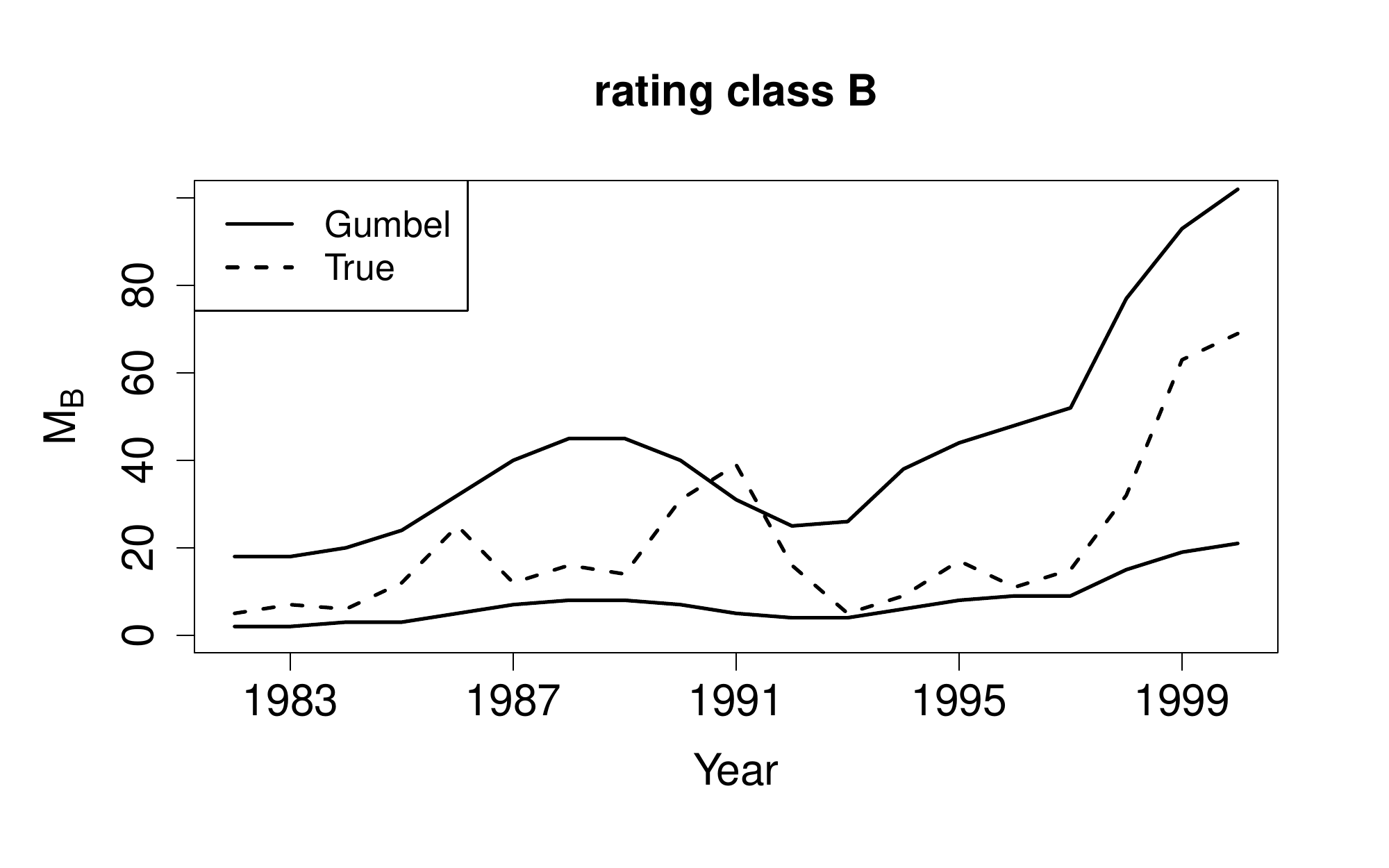}} \\ \vspace{-1cm}
\subfloat{\includegraphics[width=0.7\textwidth]{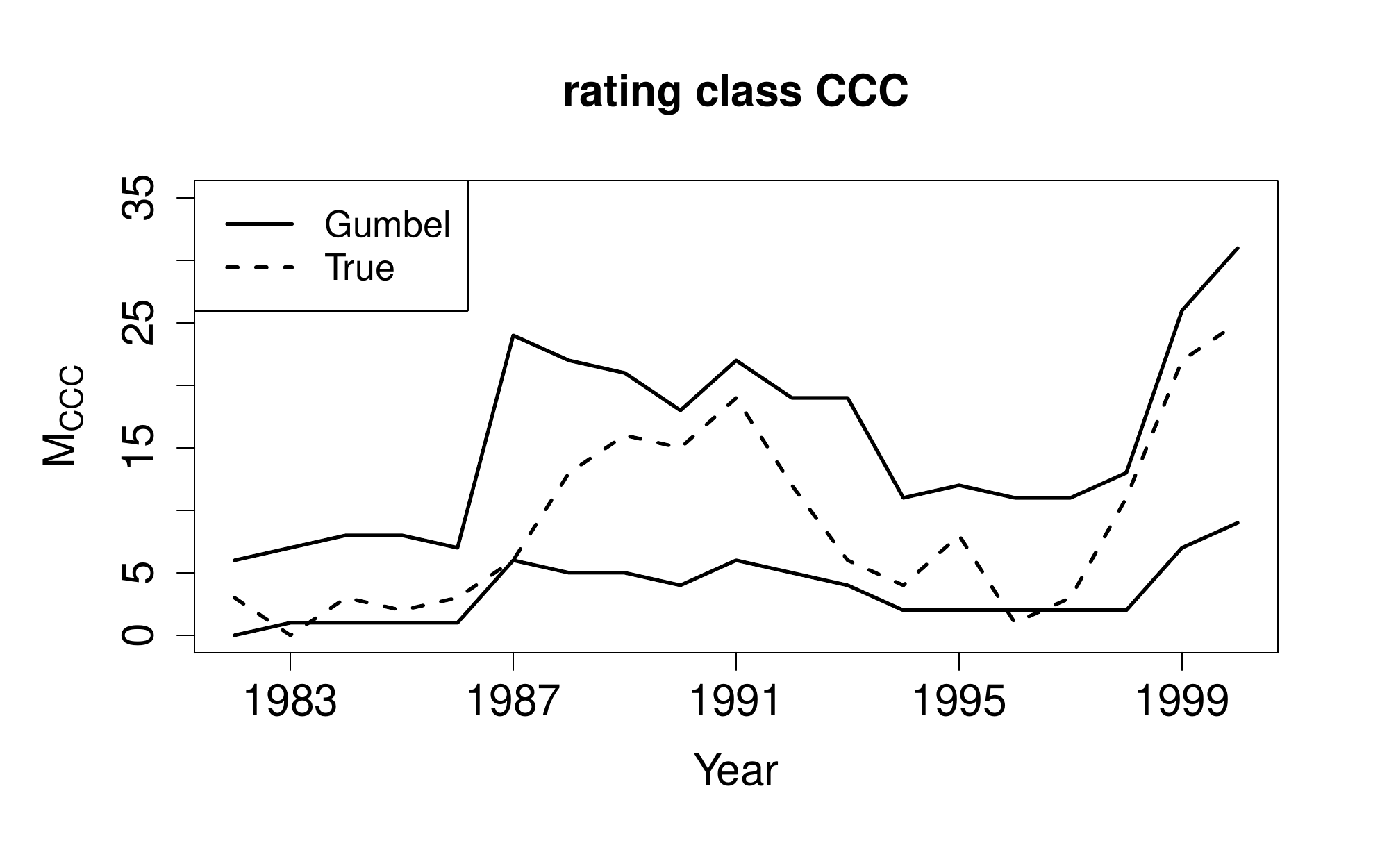}} 
\caption{\citet{sp2001} data: prediction intervals for the number of defaults obtained by simulating 5000 default matrices $M_{r,j}$ from $Q_{BB},Q_{B},Q_{CCC}$ and isolating the 4500 central observations for the Gumbel max-factor model. The dashed lines show the actual number of defaults.}
\label{fig:pred1}
\end{figure}

\begin{figure}[ht]
\centering
\subfloat{\includegraphics[width=0.7\textwidth]{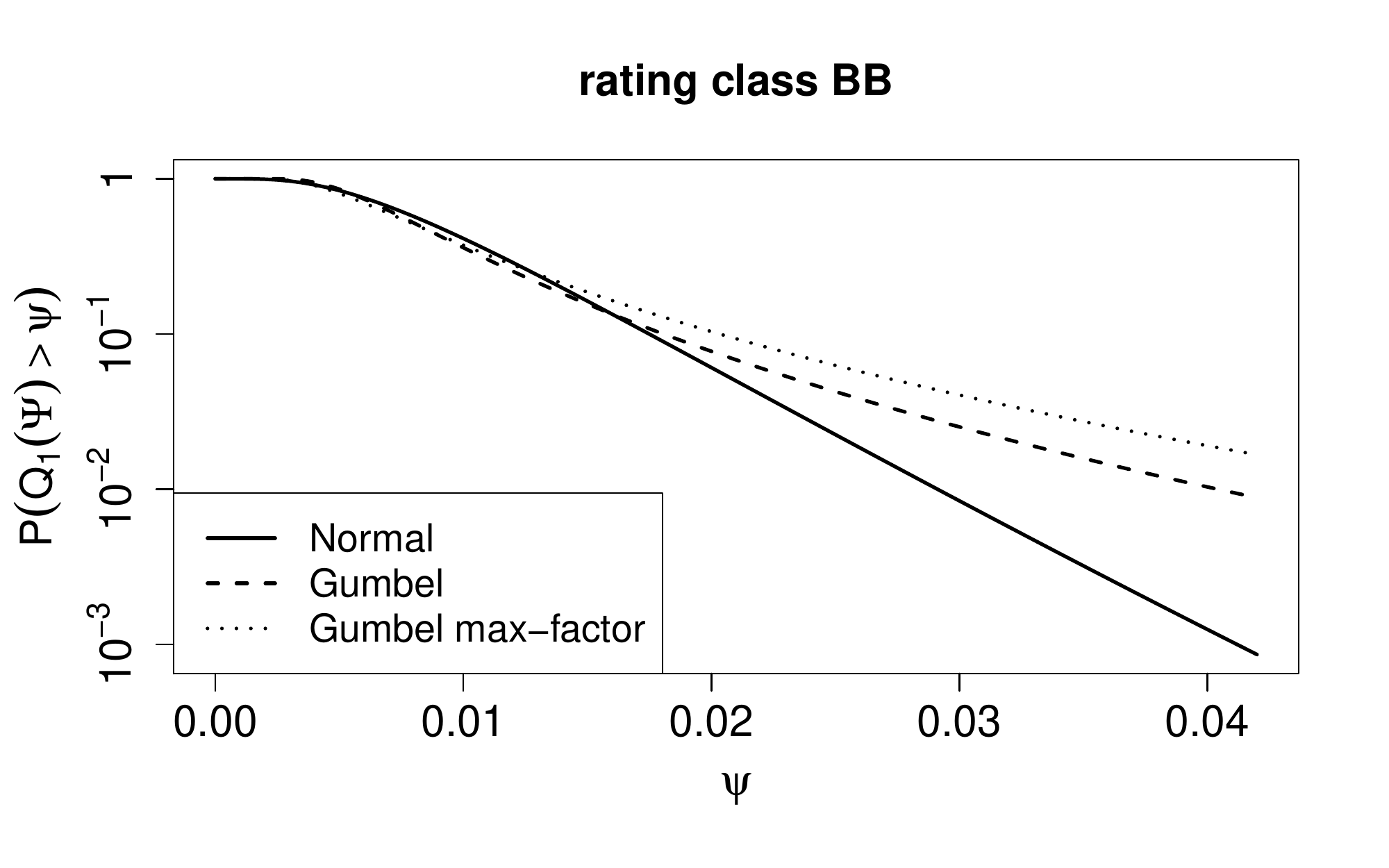}} \\ \vspace{-1cm}
\subfloat{\includegraphics[width=0.7\textwidth]{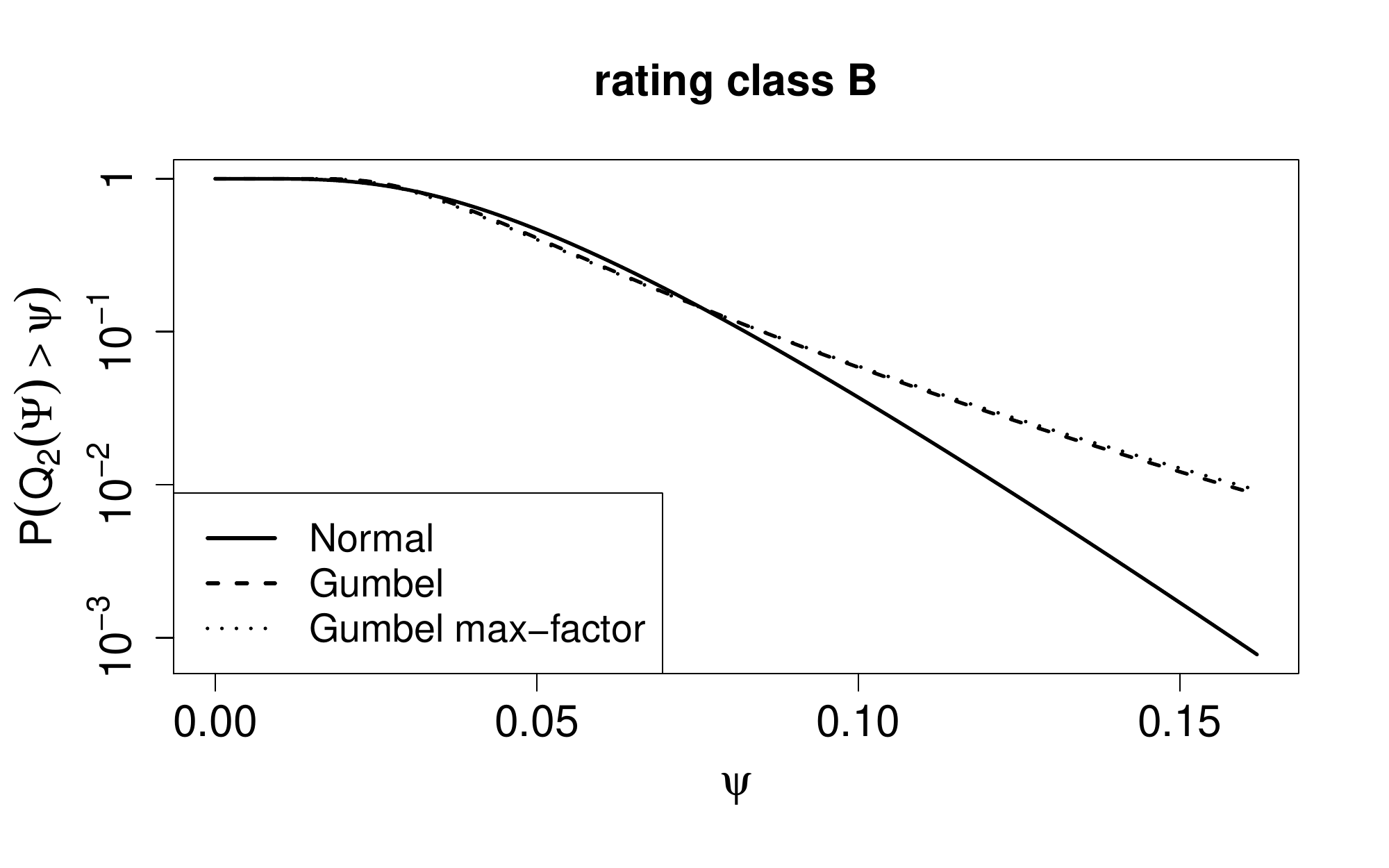}} \\ \vspace{-1cm}
\subfloat{\includegraphics[width=0.7\textwidth]{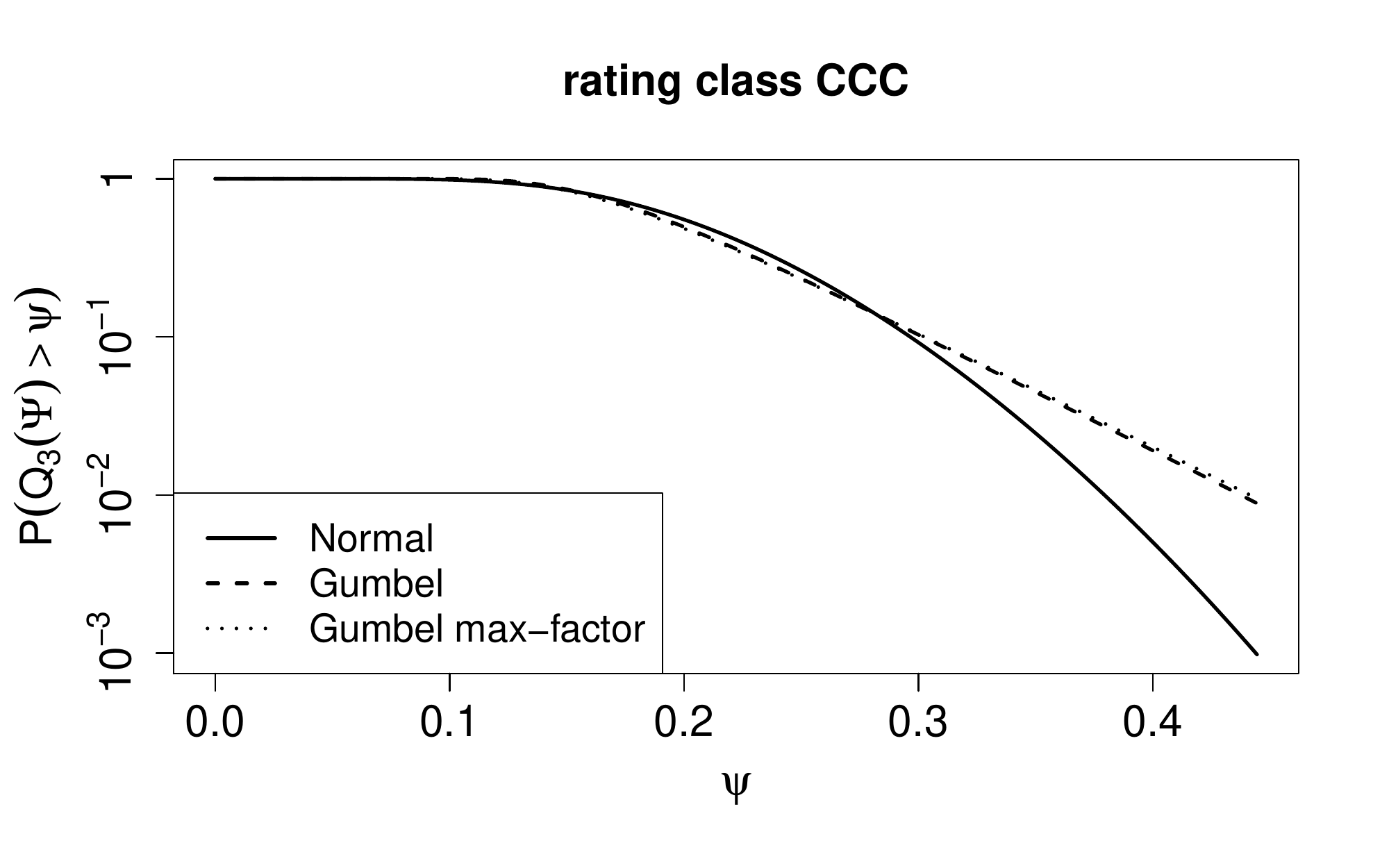}} 
\caption{\citet{sp2001} data: excess probabilities for conditional default probabilities $Q_r (\bm{\Psi})$ for $r \in \{\text{BB}, \text{B}, \text{CCC}\}$.}
\label{fig2}
\end{figure}

\clearpage

\section{Discussion}\label{sec:discuss}

In this paper, we have proposed max-factor models to account for dependencies between individual loss occurrence indicators.
Compared to the more traditional approach where the correlation is induced by linear combinations of random effects, the
max-factor specification prohibits diversification or compensation between hidden factors as only the largest effect controls
individual risk levels. The max-factor specification appears to be particularly appealing to model the occurrence of shocks
affecting policies in the portfolio, as well as the effect of common economic conditions. Compared to previous literature,
these shocks increase the conditional loss probability without systematically inducing losses on all the contracts.

This new model produces a good fit on the \citet{sp2001} credit risk data set. 
Besides classical goodness-of-fit measures based on the log-likelihood (such as AIC and BIC), we have also
proposed novel nonparametric estimators, minimizing the mean squared error, that can be used as a benchmark to evaluate
the relative merits of the different models.

The max-factor decomposition may also be interesting for credibility models decomposing the individual unobservable risk
proneness in a hierarchical way. Again, this approach is desirable in situations where no compensation is possible between the random effects associated to
the different levels but the worst case drives the individual risk proneness. We leave this topic for a future investigation.

\section*{Acknowledgements}

The authors gratefully acknowledge the financial support from the contract  ``Projet d'Actions de Recherche Concert\'ees'' No 12/17-045 of the ``Communaut\'e fran{\c c}aise de Belgique'', granted by the ``Acad\'emie universitaire Louvain'', and from IAP research network Grant P7/06 of the Belgian government (Belgian Science Policy). The second author gratefully acknowledges funding from the Belgian Fund for Scientific Research (F.R.S.-FNRS).
 
\appendix

\section{Proofs}\label{sec:proofs}


In order to establish the validity of Theorem~\ref{thmopt0}, we will need expressions for some moments of the random variables $M_{r,j}$.

\begin{Lemma}\label{moments}
For $r \in \{1,\ldots,k\}$, we have
\begin{enumerate}
\item $\EE [M_{r,j}] = m_{r,j} \pi_r$;
\item $\EE [M_{r,j} (M_{r,j} - 1)]  = m_{r,j} (m_{r,j} - 1) \pi_{rr} $;
\item $\mathrm{Var} [M_{r,j}]  = m_{r,j} (\pi_r + (m_{r,j} - 1) \pi_{rr} - m_{r,j}^2 \pi_r^2)$;
\item $\begin{aligned}[t] \mathrm{Var} [M_{r,j} (M_{r,j} - 1)]  = & \, m_{r,j} (m_{r,j} - 1) \bigl[ (2 - m_{r,j} (m_{r,j} - 1) \pi_{rr}) \pi_{rr} \\ 
& \, + \, 4 (m_{r,j} - 2) \pi^{(3)}_{r} + (m_{r,j} -2) (m_{r,j} -3) \pi^{(4)}_{r} \bigr];
\end{aligned}$
\end{enumerate}
and for $r,s \in \{1,\ldots,k\}$, $r \neq s$,
\begin{enumerate}
\item[5.] $\EE [M_{r,j} M_{s,j}]  = m_{r,j} m_{s,j} \pi_{rr}$;
\item[6.] $\begin{aligned}[t] \mathrm{Var} [M_{r,j} M_{s,j}]  = & \, m_{r,j} m_{s,j} \bigl[ (1 - m_{r,j} m_{s,j} \pi_{rs}) \pi_{rs} + (m_{s,j} - 1) \pi^{(1,2)}_{rs} \\
& \, + \, (m_{r,j} - 1) \pi^{(2,1)}_{rs} + (1 - m_{s,j} - m_{r,j} + m_{r,j} m_{s,j} \pi^{(2,2)}_{rs} ) \bigr].
\end{aligned}$ 
\end{enumerate}
\end{Lemma}

\begin{proof}[Proof of Lemma \ref{moments}]
\begin{enumerate}
\item $\EE [M_{r,j}] = \EE [ \EE [M_{r,j} \mid \Qvec_j ]] = \EE[m_{r,j} Q_{r,j}] = m_{r,j} \pi_r.$
\item $\begin{aligned}[t] \EE [M_{r,j} (M_{r,j} - 1)] & = \EE[ \EE [ M_{r,j} (M_{r,j} -1) \mid \Qvec_j ]] \\
& = \EE[m_{r,j} (m_{r,j} - 1) Q_{r,j}^2 ]\\
&  = m_{r,j} (m_{r,j} - 1) \pi_{rr},
\end{aligned}$

where the second step follows from equation \eqref{eq:binmoment}.
\item $\begin{aligned}[t] \mathrm{Var} [M_{r,j}] = \EE [M_{r,j}^2] - \EE [M_{r,j}]^2 & = 
\EE [M_{r,j} (M_{r,j} - 1) ] + \EE [M_{r,j}] - \EE[M_{r,j}]^2 \\
& = m_{r,j} (m_{r,j} - 1) \pi_{rr} + m_{r,j} \pi_r - m_{r,j}^2 \pi_r^2.
\end{aligned}$
\item We first note that
\begin{align*}
\EE [M_{r,j}^3] & = \EE [M_{r,j} (M_{r,j} - 1) (M_{r,j} - 2)] + 3 \EE[M_{r,j} (M_{r,j} - 1)] + \EE [M_{r,j}], \\
\EE [M_{r,j}^4] & = \EE [M_{r,j} (M_{r,j} - 1) (M_{r,j} - 2) (M_{r,j} -3)] + 6 \EE [M_{r,j}^3] + 7 \EE [M_{r,j} (M_{r,j} - 1)] + \EE [M_{r,j}].
\end{align*}
Then, using again equation \eqref{eq:binmoment}, 
\begin{align*} 
\mathrm{Var} [M_{r,j} (M_{r,j} - 1)] & = \EE[M_{r,j}^4] - 2 \EE[M_{r,j}^3] + \EE [M_{r,j}^2] - \EE [M_{r,j} (M_{r,j} - 1)]^2 \\
& = \EE[M_{r,j} (M_{r,j} - 1) (M_{r,j} - 2) (M_{r,j} - 3)] + 2 \EE [M_{r,j} (M_{r,j} - 1)] \\
& \qquad + 4 \EE [M_{r,j} (M_{r,j} - 1) (M_{r,j} -2 )] - \EE [M_{r,j} (M_{r,j} - 1)]^2 \\ 
& =  m_{r,j} (m_{r,j} - 1) \bigl[ (2 - m_{r,j} (m_{r,j} - 1) \pi_{rr}) \pi_{rr} \\ 
& \qquad +  4 (m_{r,j} - 2) \pi^{(3)}_{r} + (m_{r,j} -2) (m_{r,j} -3) \pi^{(4)}_{r} \bigr].
\end{align*}
\item $\EE [M_{r,j} M_{s,j}] = \EE [ \EE [M_{r,j} M_{s,j} \mid \Qvec_j ]] = \EE [m_{r,j} m_{s,j} Q_{r,j} Q_{s,j} ] = m_{r,j} m_{s,j} \pi_{rs}.$
\item $\begin{aligned}[t] \mathrm{Var} [ M_{r,j} M_{s,j}] & = \EE [ \EE [M_{r,j}^2 \mid \Qvec_j ] \, \EE [M_{s,j}^2 \mid \Qvec_j ]] - \EE [M_{r,j} M_{s,j}]^2 \\
& = m_{r,j} m_{s,j} \big[ (1 - m_{r,j} m_{s,j} \pi_{rs}) \pi_{rs} + (m_{s,j} - 1) \pi^{(1,2)}_{rs} \\
& \qquad +  (m_{r,j} - 1) \pi^{(2,1)}_{rs} + (1 - m_{s,j} - m_{r,j} + m_{r,j} m_{s,j} \pi^{(2,2)}_{rs} ) \big].
\end{aligned}$
\end{enumerate}
\end{proof}

We are now ready to proceed to the proof of the announced result.

\begin{proof}[Proof of Theorem~\ref{thmopt0}]
Write the estimators of $\pi_r$ as
\[
  \widetilde{\pi}_r (\wvec_r) = \sum_{j=1}^n w_{r,j} \widehat{Q}_{r,j}, \qquad r \in \{1,\ldots,k\}, 
\]
where $\widehat{Q}_{r,1},\ldots,\widehat{Q}_{r,n}$ have common expectation $\EE[ \widehat{Q}_{r,j} ] = \pi_r$ (Lemma~\ref{moments}, item 1) and possibly different variances $\mathrm{Var} [ \widehat{Q}_{r,j} ] = \sigma_{r,j}^2$, where
\[
  \sigma_{r,j}^2 
  = \frac{1}{m_{r,j}^2} \mathrm{Var} [ M_{r,j} ]
  =  \frac{\pi_r}{m_{r,j}} + \left( 1 - \frac{1}{m_{r,j}} \right) \pi_{rr} - \pi_r^2,
\]
by Lemma~\ref{moments} (item 3) and where the weight vector $\wvec_r = (w_{r,1}, \ldots, w_{r,n})$ has nonnegative entries. We wish to minimize the mean squared error (MSE) of $\widetilde{\pi}_r (\bm{w}_r)$ as a function of $\wvec_r$,
\begin{align*}
  \mathrm{MSE}[ \widetilde{\pi}_r(\wvec_r) ] 
  &= \mathrm{Var} [ \widetilde{\pi}_r(\wvec_r) ] + \left( \EE[ \widetilde{\pi}_r(\wvec_r) - \pi_r ] \right)^2 \\
  &= \sum_{j=1}^n w_{r,j}^2 \sigma_{r,j}^2 + \biggl( \sum_{j=1}^n w_{r,j} - 1 \biggr)^2 \mu_r^2.
  \end{align*}
Setting the partial derivatives with respect to $w_{r,1},\ldots,w_{r,n}$ equal to zero gives the solution
\begin{align*}
  w_{r,j,\mathrm{opt}} &= \frac{\sigma_{r,j}^{-2}}{\pi_r^{-2} + \sum_{t=1}^n \sigma_{r,t}^{-2}}, \qquad j = 1, \ldots, n, \\
  \mathrm{MSE}[ \widetilde{\pi}_r ( \wvec_{r, \opt} ) ] & = \frac{1}{ \pi_r^{-2} + \sum_{j=1}^n \sigma_{r,j}^{-2} },
\end{align*}
where $\wvec_{r,\opt} = (w_{r,1,\opt},\ldots,w_{r,n,\opt})$.
For the second-order probabilities $\pi_{rr}$ and $\pi_{rs}$, the same recipe can be followed. Their estimators will use the quantities $\mathrm{Var} [ M_{r,j} (M_{r,j} - 1)]$ and $\mathrm{Var} [M_{r,j} M_{s,j}]$. These are calculated in Lemma~\ref{moments} (items 4 and 6).
\end{proof}

\begin{Remark}\label{remrk}
In practice, the estimators $\widehat{\sigma}_{r,j}^2$ and $ \widehat{\sigma}_{rr,j}^2$ from Definition~\ref{thmopt} can be negative. When this happens,
a simple solution is to replace the preliminary estimator $\widehat{\pi}_{r}^{(\ell)}$ by the estimator 
\begin{equation*}
\check{\pi}_{r}^{(\ell)} = \frac{1}{n} \sum_{j=1}^n \frac{M_{r,j}^{\ell}}{m_{r,j}^{\ell}}.
\end{equation*}
Note that $\check{\pi}_r^{(\ell)} > \widehat{\pi}_r^{(\ell)}$ for $\ell > 1$.
Using $\check{\pi}_{r}^{(\ell)}$ instead of $\widehat{\pi}_r^{(\ell)}$ as a preliminary estimator in Definition~\ref{thmopt} ensures that both $\widehat{\sigma}_{r,j}^2$ and  $\widehat{\sigma}_{rr,j}^2$ are positive. Asymptotically as $m_{r,j} \to \infty$, this method is equivalent to the one described in Definition~\ref{thmopt0}.
\end{Remark}

\begin{proof}[Proof of Lemma \ref{lemma1}]
To prove Lemma~\ref{lemma1} for every $k \in \mathbb{N}$, first observe that
\begin{equation*}
 \nu_r -\sigma_r \log \left( -\log(q_r) \right) >  \mu_r -\sigma_r \log \left( -\log(q_0) \right)   \,\,\, \iff \,\,\, q_r > g_r(q_0). 
\end{equation*}
Suppose $k=1$. Then $R = \{1\}$ and
\begin{align*}
I_j & = \int_0^1 \int_0^1 f_{1,j} (q_0,q_1) \, \textrm{d}q_1 \textrm{d} q_0 \\
& =  \int_0^1 \int_0^1 \text{I} \left[q_1 \leq g_1 (q_0) \right] h_{1,j} (q_0;\mu_1) + 
\text{I} \left[ q_1 > g_1 (q_0) \right] h_{1,j} (q_1;\nu_1) \textrm{d}q_1 \textrm{d}q_0\\
& = \int_0^1 g_1(q_0) h_{1,j}(q_0 ; \mu_1) \textrm{d}q_0 + \int_0^1 \int_{g_1(q_0)}^1 h_{1,j} (q_1; \nu_1) \, \textrm{d}q_1 \textrm{d}q_0,
\end{align*}
which is equal to \eqref{eq:finallik} since $\mathcal{P}(R) = \{\emptyset, \{1\} \}$.
Next, define $R_k = \{1,\ldots,k\}$ and assume that \eqref{eq:finallik} is valid. Then for $R_{k+1} = \{1,\ldots,k+1\}$,
\begin{align*}
I_j & = \int_{[0,1]^{2}} f_{k+1,j} (q_0,q_{k+1}) \left( \int_{[0,1]^k} \left( \prod_{r=1}^k f_{r,j} (q_0,q_r) \right) \, \textrm{d}q_1 \cdots \textrm{d}q_k \right) \textrm{d}q_{k+1} \textrm{d}q_0 \\
& = \int_0^1 g_{k+1} (q_0) h_{k+1,j} (q_0; \mu_{k+1}) \left( \int_{[0,1]^k} \left( \prod_{r=1}^k f_{r,j} (q_0,q_r) \right) \, \textrm{d}q_1 \cdots \textrm{d}q_k \right) \textrm{d} q_0 \\
& \,\,\,\,\, + \int_0^1 \int_{g_{k+1} (q_0)}^1 h_{k+1,j} (q_{k+1}; \tau_{k+1}) \left( \int_{[0,1]^k} \left( \prod_{r=1}^k f_{r,j} (q_0,q_r) \right) \, \textrm{d}q_1 \cdots \textrm{d}q_k \right) \textrm{d}q_{k+1} \textrm{d} q_0 \\
& = \sum_{I \in \mathcal{P}(R_k)} \int_0^1 \left( \prod_{r \in \{R_k \setminus I\} \cup \{k+1\}} g_r (q_0) h_{r,j} (q_0; \mu_r) \right) \left(\, \prod_{r \in I} \int_{g_r (q_0)}^1 h_{r,j} (q_r; \nu_r) \textrm{d} q_r \right) \textrm{d} q_0 \\
& \,\,\,\,\, + \sum_{I \in \mathcal{P}(R_k)} \int_0^1  \left( \prod_{r \in R_k \setminus I} g_r (q_0) h_{r,j} (q_0;\mu_r) \right) \left(\, \prod_{r \in I \cup \{k+1\}} \int_{g_r (q_0)}^1 h_{r,j} (q_r; \nu_r) \textrm{d} q_r \right) \textrm{d}q_{k+1} \textrm{d} q_0 \\
& = \sum_{I \subset \mathcal{P}(R_{k+1})} \int_0^1  \left( \prod_{r \in R_{k+1} \setminus I} g_r (q_0) h_{r,j} (q; \mu_r) \right) \left(\, \prod_{r \in I} \int_{g_r (q_0)}^1 h_{r,j} (q_r; \nu_r) \textrm{d} q_r \right) \textrm{d}q_{k+1} \textrm{d} q_0.
\end{align*}
For the last step, note that if $I \in \mathcal{P}(R_{k+1})$, then either $I \in \mathcal{P}(R_k)$ so that $\{R_{k+1} \setminus I\} = \{R_k \setminus I\} \cup \{k+1\}$ 
and we get the first term on the penultimate line, or $I \notin \mathcal{P}(R_k)$ and we get the second term on the penultimate line.
\end{proof}

\section{Simulation study}\label{sec:simstudy}
In Section~\ref{sec:likelihood}, the parameter vector $\bm{\theta}$ of a factor model is estimated using maximum likelihood estimation, after which the implied marginal and joint default probabilities $\pi_r$ and $\pi_{rs}$ are obtained by plugging in the estimate of $\bm{\theta}$ in expressions~\eqref{pir} and \eqref{pirs}. In Section~\ref{sec:nonpar}, two nonparametric estimators of $\pi_r$ and $\pi_{rs}$ are introduced. Suppose that the conditional default probabilities $Q_{r,j}$ are generated from one of the multi-factor models in Section~\ref{sec:par}, i.e., model (2a) or (2b). We wish to answer the following questions.
\begin{itemize}
\item Do the weighted nonparametric estimators (Section~\ref{weight}) of $(\pi_r,\pi_{rs})$ perform better than the unweighted nonparametric estimators (Section~\ref{noweight})?
\item Does nonparametric estimation lead to better or worse estimates of $(\pi_r,\pi_{rs})$ than via maximum likelihood estimation of the parameters of the true model?
\item Does maximum likelihood estimation of the parameters of a one-factor submodel provide us with worse estimates of $(\pi_r,\pi_{rs})$ than maximum likelihood estimation of the parameters of the true model?
\item Does maximum likelihood estimation of the parameters of another multi-factor model lead to worse estimators of $(\pi_r,\pi_{rs})$ than maximum likelihood estimation of the parameters of the true model, or is the estimation quality of the (joint) default probabilities independent of the underlying data-generating process?
\end{itemize}

More specifically, we proceed as follows.
The number of risks, $m_{r,j}$, in risk category $r \in \{1, \ldots, k\}$ and time period $j \in \{1, \ldots, n\}$ is generated randomly using a beta-binomial model, for $k = 2$ and $n = 19$. The conditional default probabilities $Q_{r,j}$ are then generated using models (2a) and (2b), where the parameter values are chosen in such a way that $\pi_r$ and $\pi_{rs}$ very roughly resemble the default probabilities of the S\&P rating classes B and CCC; see Table~\ref{true}.
The quantities $\pi_r$ and $\pi_{rs}$ are then estimated by the two nonparametric estimators and by maximizing the likelihood under the assumption of one of the parametric models (1a), (2a), (1b), and (2b).  
We repeat this 1000 times and we compare the results using the relative root mean squared error (RRMSE), i.e., the root mean squared error divided by the true parameter value. 
Note that if the weighted nonparametric estimator leads to a negative value of $\widehat{\sigma}_{r,j}$ or $\widehat{\sigma}_{rr,j}$, we use the unweighted nonparametric estimator; see Remark~\ref{remrk}. This happens less than $1 \%$ of the time.

The results are presented in Tables~\ref{simstudy1} and \ref{simstudy2}. Table~\ref{simstudy1} shows the RRMSE of the estimators of the marginal default probabilities, $\pi_r$. The methods are compared in terms of the decrease, $\Delta$, of RRMSE in percent with respect to the best method. Consequently, the best method has $\Delta = 0$. When calculating $\Delta$, we take the sum of the values of the two categories.  
In terms of RRMSE, the weighted nonparametric estimator performs slightly better than the non-weighted one. Maximum likelihood estimation beats nonparametric estimation when the data are generated from model (2b), but it is the other way around when data are generated from model (2a). Misspecification of the model, for example, estimating the parameters of model (2a) although data are generated from model (2b), has no negative effect when data are generated from model (2b).

\begin{table}[t]
\centering
\begin{tabular}{lcclcc}
\toprule
\multicolumn{3}{c}{Gumbel} & \multicolumn{3}{c}{Normal} \\
\cmidrule(r){1-3}
\cmidrule(r){4-6}
\text{ } & $r=1$  & $r=2$  & \text{ } & $r=1$ & $r=2$ \\
\cmidrule(r){1-3}
\cmidrule(r){4-6}
$\mu_r$ & $-1.15$ & $-0.55$ & $\mu_r$ & $-1.60$ & $-0.85$     \\
$\nu_r$ & $-1.30$ & $-1.00$ & $\tau_r$ & $0.13$ & $0.16$  \\
$\sigma_r$ & $0.11$ & $0.15$ & $\sigma_r$ & $0.18$ & $0.28$  \\
\cmidrule(r){1-3}
\cmidrule(r){4-6}
$\pi_{r}$ & $0.0585$ & $0.2091$ & $\pi_r$ & $0.0591$ & $0.2093$  \\
\cmidrule(r){1-3}
\cmidrule(r){4-6}
$\pi_{rs} \times 100$ &  $0.397$ & $1.365$  & $\pi_{rs} \times 100$ &  $0.394$ & $1.401$  \\
 & $1.365$ & $4.759$ &  & $1.401$ & $4.980$    \\
\bottomrule
\end{tabular}
  \caption{Parameter values for the max-factor Gumbel model (left) and the sum-factor Normal model (right) used in the simulation study.}
  \label{true}
\end{table}

\begin{table}[t]
\centering
\begin{tabular}{lccclccc}
\toprule
\multicolumn{1}{l}{} &  \multicolumn{3}{c}{Gumbel} & \multicolumn{1}{l}{} &  \multicolumn{3}{c}{Normal} \\
\cmidrule(r){2-4}
\cmidrule(r){6-8}
& $r=1$ & $r=2$ & $\Delta$ & & $r=1$ & $r=2$ & $\Delta$  \\
\cmidrule(r){1-4}
\cmidrule(r){5-8}
 NP & 0.116 & 0.095 & 4 & NP & 0.110 & 0.119 & 1 \\
 NP weighted & 0.112 & 0.092 & 1 & NP weighted & 0.109 & 0.119 & 0 \\
\cmidrule(r){1-4}
\cmidrule(r){5-8}
 Model (1a) & 0.110 & 0.094 & 0 & Model (1a) & 0.112 & 0.119 & 2  \\
 Model (2a) & 0.110 & 0.094 & 0 & Model (2a) & 0.111 & 0.119 & 1  \\
 Model (1b) & 0.109 & 0.093 & 0 & Model (1b) & 0.120 & 0.121 & 6 \\
 Model (2b) & 0.110 & 0.094 & 0 & Model (2b) & 0.121 & 0.121 & 6 \\
\bottomrule
\end{tabular}
  \caption{Relative root mean squared error (RRMSE) of estimators of $\pi_r$ for data generated from a Gumbel max-factor model (left) and a Normal sum-factor model (right) with parameter values as in Table~\ref{true}. The methods are compared using $\Delta$, the increase of RRMSE in percent with respect to the best method, which has $\Delta = 0$.}
  \label{simstudy1}
\end{table}

Table~\ref{simstudy2} shows the RRMSE of the estimators of the joint default probabilities $\pi_{rs}$. Again, the weighted nonparametric estimator performs much better than the non-weighted estimator. Model misspecification has again less effect when data are generated from model (2b) than when data are generated from model (2a). Compared with the results in Table~\ref{simstudy1}, for the joint default probabilities we see a larger increase in RRMSE if we estimated the parameters of a one-factor submodel instead of a multi-factor model. 

A final thing worth noticing is that when estimating the parameters of models based on the normal distribution we obtain better estimators of low default probabilities (here $r=1$) than when using models based on the Gumbel distribution. For higher default probabilities  (here $r=2$), the quality of estimation differs less. The higher RRMSEs for the joint default probabilities based on the parameters of the Gumbel models are entirely caused by a rather high bias; while all estimators stemming from the Gumbel models exhibit (high) positive bias, the estimators stemming from the normal models are all negatively biased. Thus, although using factor models based on the normal distribution leads to estimators with lower RRMSE, it is an important drawback of models (1a) and (2a) that they are underestimating the true default probabilities.

\begin{table}[ht]
\centering
\begin{tabular}{lccclccc}
\toprule
\multicolumn{1}{l}{} & \multicolumn{3}{c}{Gumbel} & \multicolumn{1}{l}{} & 
\multicolumn{3}{c}{Normal} \\
\cmidrule(r){2-4}
\cmidrule(r){6-8}
\multicolumn{1}{l}{NP} & \multicolumn{1}{c}{$r=1$} & \multicolumn{1}{c}{$r=2$} & \multicolumn{1}{c}{$\Delta$} & \multicolumn{1}{l}{NP} & \multicolumn{1}{c}{$r=1$} & \multicolumn{1}{c}{$r=2$}  & \multicolumn{1}{c}{$\Delta$} \\
\cmidrule(r){1-4}
\cmidrule(r){5-8}
$r=1$ &  0.316 & 0.222  & 13 & $r=1$ & 0.253 & 0.207 & 7\\
$r=2$ & 0.222 & 0.211 &  & $r=2$& 0.207 & 0.255 &  \\
\cmidrule(r){1-4}
\cmidrule(r){5-8}
\multicolumn{1}{l}{NP weighted} & \multicolumn{1}{c}{$r=1$} & \multicolumn{1}{c}{$r=2$} & \multicolumn{1}{c}{$\Delta$} & \multicolumn{1}{l}{NP weighted} & \multicolumn{1}{c}{$r=1$} & \multicolumn{1}{c}{$r=2$}  & \multicolumn{1}{c}{$\Delta$} \\
\cmidrule(r){1-4}
\cmidrule(r){5-8}
$r=1$  & 0.266 & 0.202  & 0 & $r=1$  & 0.230 & 0.202 & 0 \\
 $r=2$& 0.202 & 0.197 &   & $r=2$ & 0.202 & 0.238 &  \\
\cmidrule(r){1-4}
\cmidrule(r){5-8}
\multicolumn{1}{l}{Model (1a)} & \multicolumn{1}{c}{$r=1$} & \multicolumn{1}{c}{$r=2$} & \multicolumn{1}{c}{$\Delta$} & \multicolumn{1}{l}{Model (1a)} & \multicolumn{1}{c}{$r=1$} & \multicolumn{1}{c}{$r=2$}  & \multicolumn{1}{c}{$\Delta$} \\
\cmidrule(r){1-4}
\cmidrule(r){5-8}
$r=1$  & 0.263 & 0.204  & 1 & $r=1$ & 0.263 & 0.216 & 8 \\
$r=2$ & 0.204 & 0.202 &  & $r=2$& 0.216 & 0.245 & \\
\cmidrule(r){1-4}
\cmidrule(r){5-8}
\multicolumn{1}{l}{Model (2a)} & \multicolumn{1}{c}{$r=1$} & \multicolumn{1}{c}{$r=2$} & \multicolumn{1}{c}{$\Delta$} & \multicolumn{1}{l}{Model (2a)} & \multicolumn{1}{c}{$r=1$} & \multicolumn{1}{c}{$r=2$}  & \multicolumn{1}{c}{$\Delta$} \\
\cmidrule(r){1-4}
\cmidrule(r){5-8}
$r=1$  & 0.258 & 0.204  & 0 & $r=1$  & 0.248 & 0.211 & 5 \\
$r=2$ & 0.204 & 0.202 &  &$r=2$& 0.211 & 0.243 &  \\
\cmidrule(r){1-4}
\cmidrule(r){5-8}
\multicolumn{1}{l}{Model (1b)} & \multicolumn{1}{c}{$r=1$} & \multicolumn{1}{c}{$r=2$} & \multicolumn{1}{c}{$\Delta$} & \multicolumn{1}{l}{Model (1b)} & \multicolumn{1}{c}{$r=1$} & \multicolumn{1}{c}{$r=2$}  & \multicolumn{1}{c}{$\Delta$} \\
\cmidrule(r){1-4}
\cmidrule(r){5-8}
$r=1$  & 0.289 & 0.210  & 6 &  $r=1$ & 0.397 & 0.274 & 41\\
$r=2$  & 0.210 & 0.203 &  & $r=2$ & 0.274 & 0.271 &  \\
\cmidrule(r){1-4}
\cmidrule(r){5-8}
\multicolumn{1}{l}{Model (2b)} & \multicolumn{1}{c}{$r=1$} & \multicolumn{1}{c}{$r=2$} & \multicolumn{1}{c}{$\Delta$} & \multicolumn{1}{l}{Model (2b)} & \multicolumn{1}{c}{$r=1$} & \multicolumn{1}{c}{$r=2$}  & \multicolumn{1}{c}{$\Delta$} \\
\cmidrule(r){1-4}
\cmidrule(r){5-8}
$r=1$ & 0.271 & 0.209  & 3 & $r=1$ & 0.358 & 0.259 & 32 \\
$r=2$& 0.209 & 0.203 &   & $r=2$ & 0.259 & 0.270 &  \\
\bottomrule
\end{tabular}
  \caption{Relative root mean squared error (RRMSE) of estimators of $\pi_{rs}$ for data generated from a Gumbel max-factor model (left) and a Normal sum-factor model (right) with parameter values as in Table~\ref{true}. The methods are compared using $\Delta$, the increase of RRMSE in percent with respect to the best method, which has $\Delta = 0$.}
  \label{simstudy2}
\end{table}

\clearpage

\renewcommand\refname{REFERENCES} 
\bibliographystyle{chicago} 
\bibliography{libCR}

\end{document}